\documentclass[12pt, draftclsnofoot, onecolumn]{IEEEtran}

\usepackage[font={small}]{caption}

\usepackage[dvips]{color}

\usepackage{graphicx}
\usepackage{amssymb,amsmath}
\usepackage{amsthm}
\usepackage{textcomp}
\usepackage{indentfirst}
\usepackage{url}
\usepackage[nosort]{cite}
\usepackage{color}
\usepackage{marvosym}
\usepackage{amsmath}
\usepackage{cite}
\usepackage{algorithmic}
\usepackage{algorithm}
\usepackage{epsf}
\usepackage{epstopdf}
\usepackage{balance}

\usepackage{comment}
\usepackage{todonotes}
\usepackage{epsf}
\usepackage{epsfig}
\usepackage{times}
\usepackage{epsfig}
\usepackage{graphicx}
\usepackage{bbold}
\usepackage{mathtools}
\usepackage{mathrsfs}
\usepackage{amssymb}
\usepackage{pdfpages}
\usepackage{epstopdf}
\usepackage{dsfont}
\usepackage{lettrine} 
\usepackage{amsmath,epsfig,amssymb,algorithm,amsthm,cite,url}
\usepackage{subcaption}
\allowdisplaybreaks
\usepackage{csquotes}
\usepackage{verbatim}
\usepackage[english]{babel}
\usepackage{amsmath,amssymb}

\usepackage{verbatim}

\newtheorem{theorem}{\bf Theorem}

\newtheorem{proposition}{\bf Proposition}
\newtheorem{lemma}{\bf Lemma}

\newtheorem{definition}{\bf Definition}
\newtheorem{remark}{Remark}

\newlength{\aligntop}
\newlength{\alignbot}
\makeatletter

\renewenvironment{align}{%
\vspace{\aligntop}
\start@align\@ne\st@rredfalse\m@ne
}{%
\math@cr \black@\totwidth@
\egroup
\ifingather@
\restorealignstate@
\egroup
\nonumber
\ifnum0=`{\fi\iffalse}\fi
\else
$$%
\fi
\ignorespacesafterend%
\vspace{\alignbot}\par\noindent
}

\begin{document}
	
\title{Proactive Resource Management for LTE in Unlicensed Spectrum: A Deep Learning Perspective}

\author{
\IEEEauthorblockN{Ursula Challita\IEEEauthorrefmark{1}, Li Dong\IEEEauthorrefmark{1}, and Walid Saad\IEEEauthorrefmark{2}}\\
\IEEEauthorblockA{\IEEEauthorrefmark{1}\small
School of Informatics,
The University of Edinburgh, Edinburgh, UK, \newline Emails: \{ursula.challita, li.dong\}@ed.ac.uk.}\\
\IEEEauthorblockA{\IEEEauthorrefmark{2}\small Wireless@VT, Bradley Department of Electrical and Computer Engineering, Virginia Tech, Blacksburg, VA, USA. Email: walids@vt.edu.}\vspace{-1.01cm}
\thanks{A preliminary version of this work was published at European Wireless 2017~\cite{EW_paper}.} \thanks{This research was supported by The University of Edinburgh Principal's Career Development PhD Scholarship and the U.S. National Science Foundation under Grants IIS-1633363 and CNS-1460316.}
}

	\vspace{-1cm}
\maketitle\vspace{-0.7cm}
\vspace{-0.2cm}
\begin{abstract}\vspace{-0.35cm}
LTE in unlicensed spectrum using licensed assisted access LTE (LTE-LAA) is a promising approach to overcome the wireless spectrum scarcity. However, to reap the benefits of LTE-LAA, a fair coexistence mechanism with other incumbent WiFi deployments is required. In this paper, a novel deep learning approach is proposed for modeling the resource allocation problem of LTE-LAA small base stations (SBSs). The proposed approach enables multiple SBSs to proactively perform dynamic channel selection, carrier aggregation, and fractional spectrum access while guaranteeing fairness with existing WiFi networks and other LTE-LAA operators. Adopting a proactive coexistence mechanism enables future delay-tolerant LTE-LAA data demands to be served within a given prediction window ahead of their actual arrival time thus avoiding the underutilization of the unlicensed spectrum during off-peak hours while maximizing the total served LTE-LAA traffic load. To this end, a noncooperative game model is formulated in which SBSs are modeled as Homo Egualis agents that aim at predicting a sequence of future actions and thus achieving long-term equal weighted fairness with WLAN and other LTE-LAA operators over a given time horizon. The proposed deep learning algorithm is then shown to reach a mixed-strategy Nash equilibrium (NE), when it converges. Simulation results using real data traces show that the proposed scheme can yield up to $28\%$ and $11\%$ gains over a conventional reactive approach and a proportional fair coexistence mechanism, respectively. The results also show that the proposed framework prevents WiFi performance degradation for a densely deployed LTE-LAA network.
\end{abstract}
\vspace{-0.5cm}

\IEEEpeerreviewmaketitle
\section{Introduction}
Licensed assisted access LTE (LTE-LAA) has emerged as an effective solution to overcome the scarcity of the radio spectrum \cite{intro}. Using LTE-LAA, a cellular small base station (SBS) can improve its performance by simultaneously accessing licensed and unlicensed bands. However, to achieve the promised quality-of-service (QoS) improvements from LTE-LAA, many challenges must be addressed ranging from effective co-existence with existing WiFi networks to resource allocation, multiple access, and inter-operator spectrum sharing~\cite{intro}.

If not properly deployed, LTE-LAA can significantly degrade the performance of WiFi~\cite{intro}. There has been a number of recent works~\cite{LBT_no_backoff, LBT_2, operators_mobihoc, mswim, mingzhe, Q_learning_channel, matching_channel, CU_LTE} that study the problem of enhanced LTE-LAA and WiFi coexistence. This existing body of works can be categorized into two groups: channel access~\cite{LBT_no_backoff, LBT_2, operators_mobihoc, mswim} and channel selection~\cite{Q_learning_channel, matching_channel, CU_LTE}. The authors in~\cite{LBT_no_backoff, LBT_2, operators_mobihoc} propose different channel access mechanisms based on listen-before-talk (LBT) that rely on either an exponential backoff~\cite{LBT_no_backoff}, a fixed/random contention window (CW) size~\cite{LBT_2}, or an adaptive CW size~\cite{operators_mobihoc}. Nevertheless, an exponential backoff approach leads to unnecessary retransmissions while a fixed CW size cannot handle time-varying traffic loads thus yielding unfair outcomes. The authors in~\cite{mswim} develop a holistic approach for both traffic offloading and resource sharing across the licensed and unlicensed bands but considering one SBS. In~\cite{mingzhe}, the authors study the problem of resource allocation with uplink-downlink decoupling for LTE-LAA. The authors in~\cite{dyspan} propose an inter-network coordination scheme with a centralized radio resource management for the LTE-WiFi coexistence. However, this prior art is limited to one unlicensed channel and does not jointly account for channel selection and channel access. In other words, these works do not analyze the potential gains that can be obtained upon aggregating or switching between different unlicensed channels. Operating on a fixed unlicensed channel limits the amount of cellular data traffic that can be offloaded to the unlicensed band and leads to an increase in the interference level caused to neighboring WiFi access points (WAPs) operating on that same channel.

In terms of LTE-LAA channel selection, the authors in~\cite{Q_learning_channel} propose a distributed approach based on Q-learning. A matching-based solution approach is proposed in~\cite{matching_channel}, which is both distributed and cooperative. Moreover, the work in~\cite{CU_LTE} combines channel selection along with channel access. Despite the promising results, all of these works~\cite{Q_learning_channel, matching_channel, CU_LTE} consider a reactive approach in which data requests are first initiated and, then, resources are allocated based on their corresponding delay tolerance value. Nevertheless, this sense-and-avoid approach can cause an underutilization of the spectrum due to the impulsive reconfiguration of the spectrum usage that does not account for the future dynamics of the network. Despite the predominance of the reactive LTE-WiFi coexistence solutions, cellular data traffic networks are known to exhibit statistically fluctuating and periodic demand patterns, especially applications such as file transfer, video streaming and browsing~\cite{periodic_traffic}, therefore providing an opportunity for the network to exploit the predictable behavior of the users to smooth out the traffic over time and reduce the difference between the peak and the average load. Therefore, in a \emph{proactive} approach, rather than reactively responding to incoming demands and serving them when requested, an SBS can predict traffic patterns and determine future off-peak times so that incoming traffic demand can be properly allocated over a given time window. 

Therefore, the main motivation for adopting a proactive LTE-WiFi coexistence scheme is to avoid the underutilization of the unlicensed spectrum. This is mainly accomplished by either serving a fraction of the LTE-LAA traffic when requested or shifting part of it to the future, over a given time window, so as to balance the occupancy of the unlicensed spectrum usage across time and, consequently, improve its degree of utilization. From the LTE-LAA network perspective, this will increase its transmission opportunities on the unlicensed spectrum, reduce the collision probability with WAPs and other SBSs and, hence, provide a boost for its throughput. Moreover, a proactive resource allocation scheme can exploit the inherent predictability of the future channel availability status so as to allocate resources in a window of time slots based on the predicted requests. This, in turn, can lead to a decrease in the probability of occurrence of a congestion event while ensuring a degree of fairness to the wireless local area network (WLAN). 

The main contribution of this paper is a novel deep
reinforcement learning algorithm based on long short-term memory
(RL-LSTM) cells for proactively allocating LTE-LAA resources over the
unlicensed spectrum. The LTE-LAA resource allocation problem is formulated
as a noncooperative game in which the players are the SBSs. To solve
this game,
we propose an RL-LSTM framework using which the SBSs can autonomously learn
which unlicensed channels to use along with the corresponding channel
access probability on each channel taking into account future
environmental changes, in terms of WLAN activity on the unlicensed
channels and LTE-LAA traffic loads. Unlike previous studies which are
either centralized~\cite{CU_LTE} or rely on the coordination among
SBSs~\cite{operators_mobihoc}, our approach is based on a self-organizing
proactive resource allocation scheme in which the SBSs
utilize past observations of the network state to build predictive models on spectrum
availability and to intelligently plan channel usage over a finite time
window. The use of long short term memory (LSTM) cells enables the SBSs to predict a sequence of interdependent actions over a long-term time horizon thus achieving long-term fairness among different underlying technologies. We show that, upon convergence, the proposed algorithm reaches to a mixed-strategy distribution which constitutes a mixed-strategy Nash equilibrium (NE) for the studied game. We also show that the gain of the proposed proactive resource allocation scheme and the optimal size of the prediction time window is a function of the traffic pattern of the dataset under study. To the best of our knowledge, \emph{this is the first work that exploits the framework of LSTMs for proactive resource allocation in LTE-LAA networks}. Simulation results show that the proposed approach yields significant rate improvements compared to conventional reactive solutions such as instantaneous equal weighted fairness, proportional fairness and total network throughput maximization. The results also show that the proposed scheme prevents disruption to WLAN operation in the case large number of LTE operators selfishly deploy LTE-LAA in the unlicensed spectrum. In terms of priority fairness, results show that an efficient utilization of the unlicensed spectrum is guaranteed when both technologies are given equal weighted priorities for transmission on the unlicensed spectrum.

The rest of this paper is organized as follows. In Section II, we present the system model. Section III describes the proposed coexistence game model. The LSTM-based algorithm is proposed in Section IV. In Section V, simulation results are analyzed. Finally, conclusions are drawn in Section VI.
\begin{figure}[t!]
  \begin{center}
  \centering
  \vspace{-0.4cm}
   \includegraphics[width=18.5cm]{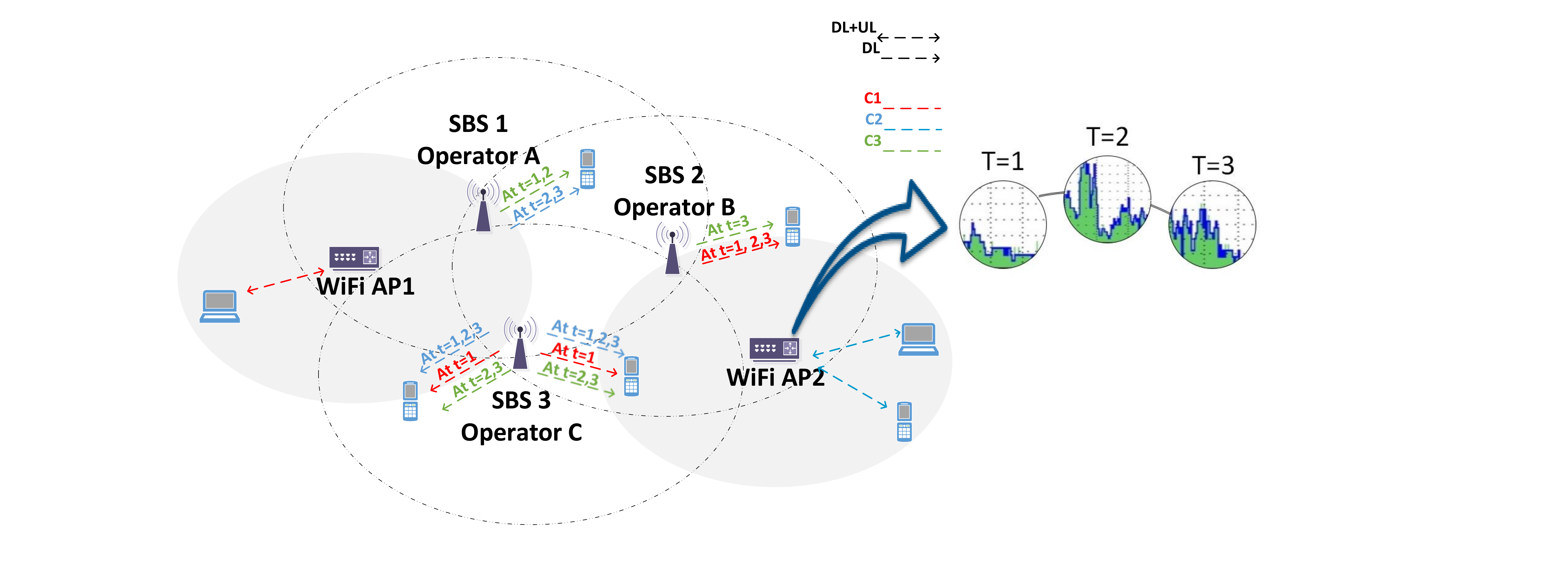} 
   \vspace{-1.2cm}
   \caption{Illustration of the system model. In the above example, $3$ SBSs belonging to different operators and $3$ unlicensed channels are only shown for simplicity. The channel selection vector over a time window of $3$ epochs is also shown.}\label{system_model}
    \vspace{-1.2cm}
  \end{center}
\end{figure}
%
\vspace{-0.4cm}
\section{System Model}
Consider the downlink of an LTE-LAA network composed of a set $\mathcal{J}$ of $J$ LTE-LAA SBSs belonging to different LTE operators, a set $\mathcal{W}$ of $W$ WAPs, and a set $\mathcal{C}$ of $C$ unlicensed channels as shown in Fig.~\ref{system_model}. Each SBS $j\in \mathcal{J}$ has a set $\mathcal{K}_j$ of $K_j$ LTE-LAA UEs associated with it. We consider a network scenario corresponding to environments such as work offices, a university campus, and airports in which the traffic load of a given WAP or SBS can be characterized through a particular model that typically remains unchanged over coarse periods of time (e.g., one day)~\cite{periodic_traffic}. We focus on the operation of the SBSs over the unlicensed band, while the licensed spectrum resources are assumed to be allocated in a conventional way~\cite{ICIC}. Both SBSs and WAPs adopt the LBT access scheme and, thus, at a particular time, a given unlicensed channel is occupied by either an SBS or a WAP. We consider the LTE carrier aggregation feature using which the SBSs can aggregate up to five component carriers belonging to the same or different operating frequency bands~\cite{LTE_aggregation}. This, in turn, would enable the SBSs to operate on multiple unlicensed channels simultaneously thus maximizing their data rate during a particular transmission opportunity.

Our goal is to jointly determine the dynamic channel selection, carrier aggregation, and fractional spectrum access for each SBS, while guaranteeing long-term airtime fairness with WLAN and other LTE-LAA operators. The main motivation for adopting a long-term fairness approach is to avoid the underutilization of the unlicensed spectrum by either serving part of the LTE-LAA traffic when requested or shifting part of it in the future over a given time window in a way that would balance the occupancy of the unlicensed spectrum usage across time and, consequently, improve its degree of utilization. This will subsequently result in an increase in the transmission opportunities for LTE-LAA as well as a decrease in the collision probability for the WLAN. To realize this, we need to dynamically analyze the usage of various unlicensed channels over a particular time window. To this end, we divide our time domain into multiple time windows of duration $T$, each of which consists of multiple time epochs $t$, as shown in Fig.~\ref{time_window}. Our objective is to proactively determine the spectrum allocation vector for each SBS at $t=0$ over $T$ while guaranteeing long-term equal weighted airtime share with WLAN. To guarantee a fair spectrum allocation among SBSs belonging to different operators, we consider inter-operator interference along with inter-technology interference. In fact, inter-operator interference is the consequence of the selfish behavior of different operators and could result in a degradation in the spectral efficiency if not managed. Next, we define $x_{j,c,t}=1$ if channel $c$ is selected by SBS $j$ during time epoch $t$, and 0, otherwise, and $\alpha_{j,c,t} \in [0,1]$. $x_{j,c,t}$ determines the channel $c$ that is used by SBS $j$ during time $t$ and $\alpha_{j,c,t}$ is the channel access probability of SBS $j$ on the unlicensed channel $c$ at time $t$.

\begin{figure}[t!]
  \begin{center}
  \centering
   \includegraphics[width=12cm]{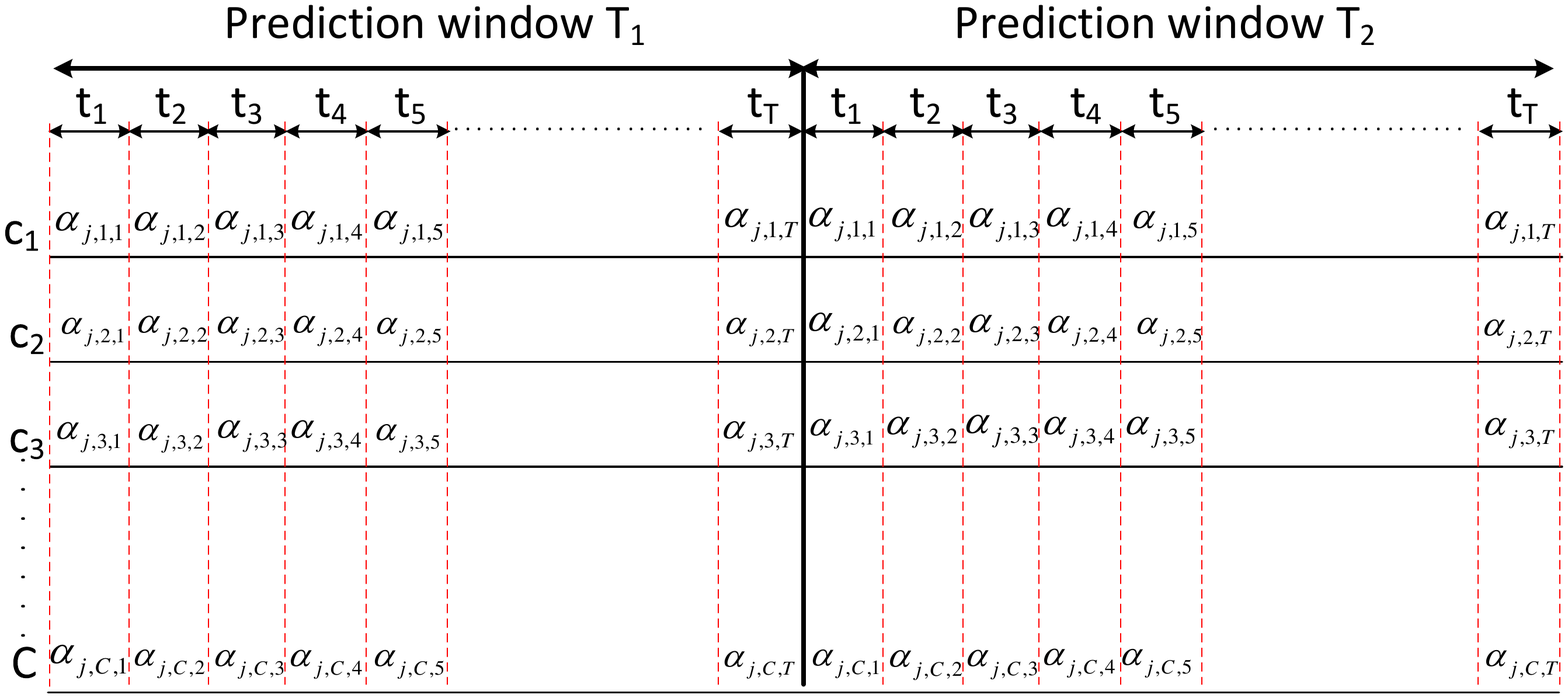} 
   \vspace{-2.8cm}
   \caption{The division of the time domain into multiple time windows $T$, each of which consists of multiple time epochs $t$.}\label{time_window}
    \vspace{-1.2cm}
  \end{center}
\end{figure}

Since the 3rd Generation Partnership Project (3GPP) has identified LBT as an access mechanism for standardizing a global solution for the operation of LTE in the unlicensed spectrum, we consider a contention-based protocol for our proposed channel access mechanism~\cite{standards_LAA}. In this protocol, prior to transmission, an SBS applies clear channel assessment for the duration of DIFS to detect the state of the channel (idle or busy) based on the detected energy level. If the channel is idle, the SBS would
backoff for a random number between [0, CW] and if the medium was still free, it gets a transmit opportunity for up to 10 LTE sub-frames (considering priority class 1 devices~\cite{standards}); it sends a reservation signal, e.g., clear-to-send (CTS), with the duration of its transmission period along with the remaining time period until the beginning of its next subframe. This allows prevention of other competing devices from getting access to the unlicensed channel until the beginning of the next subframe of the corresponding SBS and hence reserving the channel for transmission. On the other hand, if the channel was busy, the SBS keeps monitoring the channel until it becomes idle. Here, we note that our proposed algorithm is not fully compliant with the regulations in terms of CW size adjustment. In particular, we consider an exponential backoff scheme for WiFi while the SBSs adjust their CW size (and thus the channel access probability) on each of the selected channels in a way that would guarantee a long-term equal weighted fairness with WLAN and other SBSs. In essence, the exponential backoff access method that has been adopted by 3GPP for SBSs can lead to short-term unfairness~\cite{idle_sense}. This results from the fact that, after each collision, the colliding hosts double their CWs and, thus, have higher probability of choosing a larger backoff during which other hosts may benefit from channel access. This also means increased delay for hosts that doubled their CW. Therefore, the standard distributed coordination function (DCF) method controls the load on the channel by reducing the number of contending hosts, because the hosts that have failed their transmission are likely to attempt to access the channel in the future. Moreover, hosts consider all failed transmissions as collisions in DCF, however, this is not always the case. Thus, DCF bases its load control on a biased indicator, which can potentially lead to lower performance and increased unfairness~\cite{idle_sense}. On the other hand, by having a fixed CW size for each SBS during each time epoch $t$, we can alleviate these problems and, more importantly, we can decouple the load control from handling failed transmissions. It is also worth noting that small CW sizes lead to an increase in the collision probability while large CW sizes result in too much time spent waiting in idle slots. Therefore, an efficient access method should adapt the value of the CW of each SBS to the traffic conditions of the network.

To derive the throughput achieved by an LTE-LAA user equipment (UE) and a WAP, we first define the stationary probability of each WAP $w$ and each SBS $j$, $\tau_w$ and $\tau_{j,c,t}$ respectively. The stationary probability is the probability with which a given base station attempts to transmit in a randomly chosen slot. Considering an exponential backoff scheme for WiFi, the stationary probability with which WAPs transmit a packet during a particular WiFi time slot, $\tau_w$, will be given by~\cite{bianchi}:
\vspace{0.3cm}
\begin{align}
\tau_w=\frac{2(1-2q_w)}{(1-2q_w)(\mathrm{CW_{min}}+1)+q_w \mathrm{CW_{min}} (1-(2q_w)^m)},
\end{align}
\noindent where $q_w$ is the collision probability of a WAP, $m$ is the maximum backoff stage where $\mathrm{CW_{max}}=2^m \mathrm{CW_{min}}$. $\mathrm{CW_{min}}$ and $\mathrm{CW_{max}}$ are the minimum and maximum contention window size, respectively. For LTE-LAA, $m$=0 since no exponential backoff is considered, and, thus, the stationary probability of an SBS on a given unlicensed channel $c$ during time epoch $t$ will be $\tau_{j,c,t}=\frac{2}{\mathrm{CW}_{j,c,t}+1}$, where $\mathrm{CW}_{j,c,t}$ is the contention window size of SBS $j$ on channel $c$ during time epoch $t$. Therefore, we do not consider a contention stage for LTE-LAA and, thus, the CW size of the SBSs is not doubled after each unsuccessful transmission. Instead, the SBSs consider a fixed value for the CW for each time epoch $t$ and this value is adjusted adaptively from one time epoch $t$ to another in order to control the corresponding channel access probabilities over the unlicensed band for different time epochs. The collision probability of a WAP is defined as $q_w=1-\prod_{v=1, v\neq w}^{W}(1-\tau_v)\prod_{j=1}^{J}(1-\tau_{j,c,t})$, where $c$ is the channel used by WAP $w$. The throughput $R_{w}$ of a WAP $w$ during a particular WiFi time slot will be:
\vspace{0.3cm}
\begin{align}\label{wifi_throughput}
R_{w}=\frac{P_{w,\textrm{succ}}\cdot E[D_w]}{P_{w,\textrm{idle}} \cdot \theta + P_{w,\textrm{busy}} \cdot T_b},
\end{align}
where $E[D_w]$ is the expected payload size for WAP $w$, $P_{w,\textrm{succ}}=\tau_w \prod_{v=1, v\neq w}^{W} (1-\tau_v)\prod_{j=1}^{J} (1-\tau_{j,c,t})$ is the probability of a successful transmission, $P_{w,\textrm{idle}}=\prod_{j=1}^{J}(1-\tau_{j,c,t})\prod_{w=1}^{W}(1-\tau_w)$
is the probability of an idle slot, and $P_{w,\textrm{busy}}=1-\prod_{j=1}^{J}(1-\tau_{j,c,t})\prod_{w=1}^{W}(1-\tau_w)$ is the probability of a busy slot, regardless of whether it corresponds to a collision or a successful transmission. $\theta$ and $T_b$ are, respectively, the average durations of an idle and a busy slot and, thus, the denominator in (\ref{wifi_throughput}) corresponds to the mean duration of a WiFi medium access control (MAC) slot.

The achievable airtime fraction for an LTE-LAA SBS $j$ on channel $c$ at time $t$ is:
\vspace{0.2cm}
\begin{align}
\alpha_{j,c,t}=\tau_{j,c,t} \prod_{i=1, i\neq j}^{J} (1-\tau_{i,c,t})\prod_{w=1}^{W} (1-\tau_w).
\end{align}

\begin{figure}[t!]
  \begin{center}
  \centering
   \includegraphics[scale=0.4]{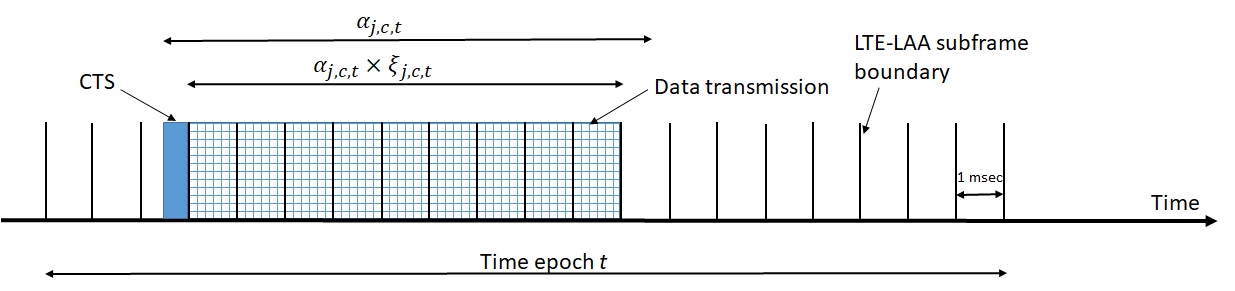} 
   \vspace{-0.31cm}
   \caption{An illustrative example for computing the actual data transmission time of an LTE-LAA SBS.}\label{illustration}
    \vspace{-0.9cm}
  \end{center}
\end{figure}

The airtime fraction represents the time allocated for an SBS on channel $c$ during time $t$, which essentially accounts for both the data transmission time as well as the reservation signal overhead. Here, it is important to account for the reservation signal overhead during a transmission burst of an SBS when computing the throughput, as done in~\cite{choi}. As such, we let $\xi_{j,c,t}$ be the average fraction of time of $\alpha_{j,c,t}$ during which LTE-LAA SBS is transmitting data. Fig.~\ref{illustration} provides an illustrative example for computing the fraction of time of $t$ during which LTE-LAA SBS is transmitting data. Thus, the total throughput of all $K_{j,t}$ UEs that are served by SBS $j$ during time epoch $t$ is:
\begin{align}
R_{j,t}= \sum_{c=1}^C \alpha_{j,c,t} \xi_{j,c,t} r_{j,c,t},
\end{align}
where
\begin{align}~\label{rate_expression}
r_{j,c,t}= \sum_{k=1}^{K_{j,t}} B_c \mathrm{log}\Big(1+\frac{P_{j,c,t} h_{j,k,c,t}}{I_{j,c,t}+ B_c N_0}\Big).
\end{align}

\vspace{0.2cm}\noindent Here, $I_{j,c,t}=\sum_{w=1}^W P_{w,c,t} h_{w,k,c,t} + \sum_{i=1, i\neq j}^J P_{i,c,t} h_{i,k,c,t}$ is the interference level on SBS $j$ when operating on channel $c$ during time $t$ and $B_c$ is the bandwidth of channel $c$. $P_{j,c,t}$ is the transmit power of SBS $j$ on channel $c$ during time $t$. $h_{j,k,c,t}$ is the channel gain between SBS $j$ and UE $k$ on channel $c$ during time $t$. $N_0$ is the power spectral density of additive white Gaussian noise. Since SBSs and WAPs both adopt LBT, then one cell may occupy the entire channel at a given time thus transmitting \emph{exclusively} on a given channel $c$. However, hidden and exposed terminals could be present on a given channel which can result in interference and thus a degradation in the throughput.

Given this system model, next, we develop an effective spectrum allocation scheme that can allocate the appropriate unlicensed channels along with the corresponding channel access probabilities to each SBS simultaneously over $T$, at $t=0$.
\vspace{-0.4cm}
\section{Proactive Resource Allocation Scheme for Unlicensed LTE}
\subsection{Proactive Resource Allocation Game}
We formulate the resource allocation problem as a noncooperative game $\mathcal{G}\textrm{=}(\mathcal{J}, \mathcal{A}_j, u_j)$ with the SBSs in $\mathcal{J}$ being the players, each of which must choose a channel selection and channel access pair $a_{j,c,t}=$($x_{j,c,t}$,$\alpha_{j,c,t}$) $\in \mathcal{A}_j$ at $t=0$ for each $t$ of the next time window $T$. The objective of each SBS $j$ is to maximize its total throughput over the selected channels:


\vspace{-0.4cm}
\begin{align}\label{utility}
u_j(\boldsymbol{a}_j, \boldsymbol{a}_{-j})= \sum_{t=1}^T \sum_{c=1}^{C} \alpha_{j,c,t} \xi_{j,c,t} r_{j,c,t},
\end{align}

\noindent where $\boldsymbol{a}_{j} = [(a_{j,1,1}, \cdots, a_{j,1,T}), \cdots, (a_{j,C,1}, \cdots, a_{j,C,T})]$ and $\boldsymbol{a}_{-j}$ correspond, respectively to the action vector of SBS $j$ and all other SBSs, over all the channels $\mathcal{C}$ during $T$. Note that the utility function (\ref{utility}) of SBS $j$ depends on its actions as well as those of other SBSs which makes the formulation of a game model suitable for this problem. This is mainly due to the interference from other SBSs transmitting on the same channel as SBS $j$ as it was shown previously in the definition of the rate expression in (\ref{rate_expression}). The goal of each SBS $j$ is to maximize its own utility:
\begin{align}\label{obj}
\max_{\mathbf{\boldsymbol{\emph{a}}_\emph{j} \in \mathcal{A}_\emph{j}}} u_{j}(\boldsymbol{a}_j, \boldsymbol{a}_{-j})  \;\; \forall j\in \mathcal{J},
\end{align}
\vspace{-0.8cm}
\begin{align}\label{cons_1}
\hspace{-2cm}\mathrm{s. t.} \;\;\;\;\;\;\;\;\;\;\; \alpha_{j,c,t} \leq x_{j,c,t} \;\;\;\forall c, t,
\end{align}
\vspace{-0.8cm}
\begin{align}\label{cons_2}
\sum_{c=1}^{C} x_{j,c,t} \leq \mathrm{min} (M_c, C)  \;\;\;\forall t,
\end{align}
\vspace{-0.4cm}
\begin{align}\label{cons_3}
\sum_{t_T=1}^t\sum_{c=1}^{C} \alpha_{j,c,t_T} B_c \leq \sum_{t_T=1}^tf(L_{j,t_T}) \;\;\;\forall t,
\end{align}
\vspace{-0.5cm}
\begin{align}\label{cons_4}
\alpha_{w,c,t} + \alpha_{j,c,t} + \sum_{i=1, i\neq j}^J \alpha_{i,c,t} \leq t_{\mathrm{max}} \;\;\;\forall c, t,
\end{align}
\vspace{-0.7cm}
\begin{align}\label{cons_5}
x_{j,c,t}\in \{0,1\}, \;\; \alpha_{j,c,t}\in [0,1] \;\;\;\forall c, t,
\end{align}
\vspace{-0.1cm}
where $M_c$ is the total number of unlicensed channels which an SBS can aggregate. (\ref{cons_1}) allows the allocation of a channel
access proportion for SBS $j$ on channel $c$ during $t$ only if SBS
$j$ transmits on channel $c$ at time $t$. (\ref{cons_2}) guarantees that each SBS can aggregate a maximum of $M_c$ channels at a given time $t$. (\ref{cons_3}) limits the amount of allocated bandwidth to the required demand where $f(L_{j,t})$ captures the relationship between bandwidth requirement and offered load.
(\ref{cons_4}) captures coupling constraints which limit the proportion of time used by SBSs and WLAN on a given unlicensed band to the maximum fraction of time an unlicensed channel can be used, $t_{\mathrm{max}}$\footnote{$t_{\mathrm{max}}$ depends on the channel access method in the unlicensed band and should
be strictly less than 1 in the case of LBT, otherwise, the channel will always be sensed busy and devices would not be able to access it.}. (\ref{cons_5}) represents the feasibility constraints.


Given the fact that different operators and technologies have equal priorities on the unlicensed spectrum, we incorporate the Homo Egualis (HE) anthropological model, an inequity-averse based fairness model, into the strategy design of the agents~\cite{gintis}.
\vspace{-0.1cm}
\begin{definition} Inequity aversion \emph{is the preference for fairness and resistance to incidental inequalities. In other words, it refers to the willingness of giving up some material payoff in order to move in the direction of more equitable outcomes.}
\end{definition}
In an HE society, agents focus not only on maximizing their own payoffs, but also become aware of how their payoffs are compared to other agents' payoffs~\cite{gintis, inequity}. Therefore, their utility function is influenced not only by their own reward, but also by envy
and altruism. An agent is altruistic to others if its payoff is above an equitable benchmark and is envious of the others if its payoff exceeds that benchmark and therefore, an unfair distribution of resources among agents results in disutility for inequity-averse agents. The HE concept comes from the anthropological literature in which Homo sapiens evolved in small hunter-gatherer groups without a centralized governance~\cite{gintis}.

In fact, we incorporate the notion of airtime fairness in the modeling of our HE agents. The average airtime per radio system is considered as one of the most important fairness metrics in the unlicensed band and is the focus of this work~\cite{mangold_HE}. Our motivation for considering a time-fair channel allocation scheme is to overcome the rate anomaly problem that arises when different nodes use distinct data rates, which leads to the slowest link limiting the system performance~\cite{operators_mobihoc, mangold_HE}, and~\cite{idle_sense}. Therefore, to model our players as HE agents, we consider the following two coupling constraints for the allocated airtime fraction on each channel $c$ for each SBS $j$:

\vspace{-0.4cm}
\begin{align}\label{cons_6}
\small
\frac{1}{w_{j,c}} \frac{1}{T} \frac{\sum_{t=1}^T\alpha_{j,c,t}}{\sum_{t=1}^T\bar{L}_{j,t}}\textrm{$=$}\frac{1}{w_{i,c}}\frac{1}{T} \frac{\sum_{t=1}^T \alpha_{i,c,t}}{\sum_{t=1}^T \bar{L}_{i,t}} \;\forall c\in \mathcal{\widehat{C}}_{j}, i\in \mathcal{\widehat{S}}_{j,c} (i\neq j)\textrm{,}
\end{align}
\vspace{-0.3cm}
\begin{align}\label{cons_7}
\small
\frac{1}{T} \frac{\sum_{t=1}^T\sum_{n \in\mathcal{S}_{c,t}} \alpha_{n,c,t}}{P_{\mathrm{LTE}}\sum_{t=1}^T \sum_{n\in \mathcal{S}_{c,t}} \bar{L}_{n,t}}\textrm{$=$}\frac{1}{T}  \frac{\sum_{t=1}^T\alpha_{w,c,t}}{P_{\mathrm{WiFi}}\sum_{t=1}^T L_{w,c,t}} \;\;\forall c\in \mathcal{\widehat{C}}_{j}\textrm{,}
\end{align}

\vspace{0.1cm}\noindent where $\mathcal{\widehat{C}}_{j}$ is the subset of channels used by SBS $j$ during $T$. $\mathcal{S}_{c,t}$ is the subset of SBSs that are transmitting over channel $c$, $c \in \mathcal{\widehat{C}}_{j}$, during time $t$ and $\mathcal{\widehat{S}}_{j,c}$ is the subset of other neighboring SBSs, $i \neq j$, that are using the same channel $c \in \mathcal{\widehat{C}}_{j}$ as SBS $j$ during $T$. $\bar{L}_{j,t}=L_{j,t}-\sum_{c'}f(\alpha_{j,c',t})$ corresponds to the remaining traffic that needs to be served by SBSs $j$ with $L_{j,t}$ being the $\emph{total}$ aggregate traffic demand of SBS $j$ on channel $c$ during time epoch $t$. $f(.)$ corresponds to the served traffic load as a function of the airtime allocation. $c'$ represents all the set of channels except channel $c$. $\alpha_{w,c,t}= \mathrm{min}(f(L_{w,c,t})$, $t_{\mathrm{max}}-\alpha_{j,c,t} - \sum_{i \in\mathcal{S}_{j,c,t}} \alpha_{i,c,t}$) is the airtime allocated for WLAN transmissions over channel $c$ during time $t$. $P_{\mathrm{WiFi}}$ and $P_{\mathrm{LTE}}$ correspond to the priority metric defined for each technology when operating on the unlicensed band. These parameters allow adaptation of the level of fairness between LTE-LAA and WLAN.

Constraint (\ref{cons_6}) represents inter-operator fairness which guarantees an equal weighted airtime allocation among SBSs belonging to different operators on a given channel $c$. The adopted notion of fairness is based on a long-term weighted equality over $T$, as opposed to instantaneous weighted equality. $w_{j,c}=\sum_{t=1}^T x_{j,c,t}$ is the weight of SBS $j$ on channel $c$ during $T$ and thus different SBSs are assigned different weights on each channel $c$ based on the number of time epochs $t$ a given SBS $j$ uses that particular channel. (\ref{cons_7}) defines an inter-technology fairness metric to guarantee a long-term equal weighted airtime allocation over $T$ for both LTE-LAA and WiFi. Therefore, (\ref{cons_6}) and (\ref{cons_7}) reflect the inequity aversion property of the SBSs.

In fact, the optimal value of $T$, which corresponds to the time window size that allows the maximum achievable throughput for LTE-LAA as compared to the reactive approach, is dataset dependent. Next, we characterize the optimal value of $T$ under a uniform traffic distribution.


\begin{proposition}~\label{proposition_traffic}
\emph{For a uniform traffic distribution, the optimal value of $T$ is equal to 1.}
\end{proposition}


\begin{proof}
Under a uniform demand model, the traffic load for each of SBS $j$ and WAP $w$ is an independent and identically distributed (i.i.d.) sequence of random variables
which implies that all requests of the same user are statistically indistinguishable over time. In our model, WAPs are passive in that their channel selection action is fixed and, thus, the activity on a given channel is characterized by the level of activity of WAPs operating on that channel. In that case, the WLAN traffic load on each channel also
follows a uniform distribution. At the convergence point, (\ref{cons_1})-(\ref{cons_7}) are satisfied and, hence, the average airtime allocated to the LTE-LAA
network on channel $c$ over the time window $T$ will be:
\vspace{0.15cm}
\begin{align}\label{long_term_1}
\frac{1}{T} \sum_{t=1}^T \sum_{j \in\mathcal{S}_{c,t}} \alpha_{j,c,t}=\frac{P_{\mathrm{LTE}}}{P_{\mathrm{WiFi}}} \frac{\sum_{t=1}^T \sum_{j \in\mathcal{S}_{c,t}} \bar{L}_{j,t}}{\sum_{t=1}^T L_{w,c,t}} \frac{1}{T} \sum_{t=1}^T \alpha_{w,c,t} \;\;\forall c\in \mathcal{C}\textrm{,}
\end{align}


\vspace{0.1cm}However, for the case of uniform traffic demand, the channel selection
vector over $T$ is the same for each SBS because the network state is the same for every $t$ in $T$. Moreover, if an SBS aggregates multiple channels, then its load on each channel is the same for each $t$
in $T$. This implies that $\bar{L}_{j,t}$ for each SBS $j$ is uniform over $T$
and thus $\frac{\sum_{t=1}^T \sum_{j \in\mathcal{S}_{c,t}} \bar{L}_{j,t}}{\sum_{t=1}^T L_{w,c,t}}=\frac{\sum_{j \in\mathcal{S}_{c,t}} \bar{L}_{j,t}}{L_{w,c,t}}$.
Consequently, (\ref{long_term_1}) can be written as:
\vspace{0.1cm}
\begin{align}\label{long_term_2}
\frac{1}{T} \sum_{t=1}^T \sum_{j \in\mathcal{S}_{c,t}} \alpha_{j,c,t}=\frac{P_{\mathrm{LTE}}}{P_{\mathrm{WiFi}}} \frac{\sum_{j \in\mathcal{S}_{c,t}} \bar{L}_{j,t}}{L_{w,c,t}} \frac{1}{T} \sum_{t=1}^T \alpha_{w,c,t}\;\;\forall c\in \mathcal{C}\textrm{,}
\end{align}

\vspace{0.1cm}When $T=1$, the airtime allocated to the LTE-LAA network on channel $c$ will be:
\vspace{0.2cm}
\begin{align}
\sum_{j \in\mathcal{S}_{c,t}}\alpha_{j,c,t}=\frac{P_{\mathrm{LTE}}}{P_{\mathrm{WiFi}}} \frac{\sum_{j \in\mathcal{S}_{c,t}} \bar{L}_{j,t}}{L_{w,c,t}} \alpha_{w,c,t} \;\;\; \forall t, c\in \mathcal{C}\textrm{,}
\end{align}

Over a fixed time window $T$, the average airtime allocated to the LTE-LAA network on channel $c$ can be written as:
\begin{align}\label{instantaneous}
\frac{1}{T} \sum_{t=1}^T \sum_{j \in\mathcal{S}_{c,t}} \alpha_{j,c,t}&=\frac{1}{T} \sum_{t=1}^T \Big(\frac{P_{\mathrm{LTE}}}{P_{\mathrm{WiFi}}} \frac{\sum_{j \in\mathcal{S}_{c,t}} \bar{L}_{j,t}}{L_{w,c,t}} \alpha_{w,c,t}\Big) \\
&= \frac{P_{\mathrm{LTE}}}{P_{\mathrm{WiFi}}} \frac{\sum_{j \in\mathcal{S}_{c,t}} \bar{L}_{j,t}}{L_{w,c,t}} \frac{1}{T} \sum_{t=1}^T \alpha_{w,c,t}\;\;\forall c\in \mathcal{C}. \label{instantaneous_2}
\end{align}
(\ref{instantaneous_2}) is equivalent to (\ref{long_term_2}) and, hence, our proposed framework does not offer any gain for the LTE-LAA network
when considering a time window $T$ larger than 1 in the case of a uniform traffic pattern. This completes the proof.
\end{proof}

From Proposition~\ref{proposition_traffic}, we conclude that the gain of our proposed long-term fairness notion is evident in the case of traffic fluctuations. Under a uniform traffic distribution, the SBSs
cannot make use of future off-peak times to shift part of their traffic forward in time and, hence, the gain is limited to predicting the network state for the next time epoch only.
It is also worth noting that the gain of the proactive scheduling approach decreases in the case of a highly congested WLAN network. This is mainly due to the fact
that the system becomes more congested with incoming requests, thereby restricting the opportunities of shifting part of the LTE-LAA load in the future.

\subsection{Equilibrium Analysis}
Our game $\mathcal{G}$ is a generalized Nash equilibrium problem (GNEP) in which both the objective functions and the action spaces are coupled. To solve the GNEP, we incorporate the Lagrangian penalty method into the utility functions thus reducing it to a simpler Nash equilibrium problem (NEP). The resulting penalized utility function will be, $\forall (j\in \mathcal{J})$:
\begin{align}\label{penalized_utility}
\small
\widehat{u}_j(\boldsymbol{a}_j, \boldsymbol{a}_{-j})= \sum_{t=1}^T \sum_{c=1}^C \alpha_{j,c,t} r_{j,c,t}
\small
&\mathrm{-}\rho_{1,j} \sum_{c=1}^C \sum_{t=1}^T \Big(\mathrm{min} (0,t_{\mathrm{max}} - \alpha_{w,c,t} - \alpha_{j,c,t} - \sum_{i=1, i\neq j}^J\alpha_{i,c,t})\Big)^2 \nonumber\\
\small
&\mathrm{-}\rho_{2,j}\hspace{-0.1cm} \sum_{c\in \mathcal{\widehat{C}}_{j}} \sum_{i\in \mathcal{\widehat{S}}_{j,c} (i\neq j)} \hspace{-0.1cm}\frac{1}{T^2}\left(\hspace{-0.1cm}\frac{1}{w_{j,c}} \frac{\sum_{t=1}^T\alpha_{j,c,t}}{\sum_{t=1}^T\bar{L}_{j,t}}\mathrm{-}\frac{1}{w_{i,c}} \frac{\sum_{t=1}^T \alpha_{i,c,t}}{\sum_{t=1}^T \bar{L}_{i,t}} \hspace{-0.1cm}\right)^2 \vspace{-0.3cm} \nonumber \\ \vspace{-0.3cm}
\small
&\mathrm{-} \rho_{3,j} \hspace{-0.1cm}\sum_{c\in \mathcal{\widehat{C}}_{j}} \frac{1}{T^2}\hspace{-0.1cm}\left(\hspace{-0.1cm} \frac{\sum_{t=1}^T\sum_{n \in\mathcal{S}_{c,t}} \alpha_{n,c,t}}{P_{\mathrm{LTE}}\sum_{t=1}^T \sum_{n\in \mathcal{S}_{c,t}} \bar{L}_{n,t}}\mathrm{-} \frac{\sum_{t=1}^T\alpha_{w,c,t}}{P_{\mathrm{WiFi}}\sum_{t=1}^T L_{w,c,t}}\hspace{-0.1cm}\right)^\textrm{$2$}\hspace{-0.1cm}\textrm{,}
\normalsize
\end{align}

\vspace{0.3cm}
\noindent where $\rho_{1,j}$, $\rho_{2,j}$ and $\rho_{3,j}$ are positive penalty coefficients corresponding to constraints (\ref{cons_4}), (\ref{cons_6}), and (\ref{cons_7}), respectively.
Here, we consider equal penalty coefficients for all players for each coupled constraint, $\rho_{1,j}=\rho_1$, $\rho_{2,j}=\rho_2$ and $\rho_{3,j}=\rho_3$. This allows all SBSs to have equal incentives to give up some payoff in order to satisfy the coupled constraints. To determine the values of $\rho_1$, $\rho_2$ and $\rho_3$, we adopt the incremental penalty algorithm in~\cite{fukushima} that guarantees the existence of penalty parameters $\boldsymbol{\rho}_l^*=[\rho_1^*, \rho_2^*, \rho_3^*]$ that satisfy the coupled constraints.

In our game $\mathcal{G}$, $\alpha_{j,c,t}$ is a continuous variable bounded between 0 and 1, however, for a particular network state, we are interested only in a certain region of the continuous space where the optimal actions are expected to be. Therefore, we will propose a sampling-based approach to discretize $\alpha_{j,c,t}$ in Section~\ref{learning}. Under such a discretization of the action space, we turn our attention to mixed strategies in which players choose their strategies probabilistically. Such a mixed-strategy approach enables us to analyze the frequency with which players choose different channels and channel access combinations. In fact, the optimal policy is often stochastic and therefore requires the selection of different actions with specific probabilities~\cite{policy_gradient}. This, in turn, validates our choice of adopting a mixed strategy approach as opposed to a pure strategy one that is oriented towards finding deterministic policies. A player can possibly choose different possible actions with different probabilities which enables it to play a combination of strategies over time. Moreover, unlike pure strategies that might not exist for a particular game, there always exists at least one equilibrium in mixed strategies~\cite{nash_existence}.

Let $\Delta(\mathcal{A})$ be the set of all probability distributions over the action space $\mathcal{A}$ and $\boldsymbol{p}_j=[p_{j,\boldsymbol{a}_{1}} \cdots, p_{j,\boldsymbol{a}_{\mid\mathcal{A}_j\mid}}]$ be a probability distribution with which SBS $j$ selects a particular action from $\mathcal{A}_j$. Therefore, our objective is to maximize the expected value of the utility function, $\overline{u}_j(\boldsymbol{p}_j, \boldsymbol{p}_{-j})={\mathds{E}_{\boldsymbol{p}_{j}} \left[\widehat{u}_{j}\left( \boldsymbol{a}_{j},\boldsymbol{a}_{-j} \right)\right]}=\sum_{\boldsymbol{a}\in\mathcal{A}}\widehat{u}_{j}(\boldsymbol{a}_j, \boldsymbol{a}_{-j}) \prod_{j=1}^J p_{j,\boldsymbol{a}_j}$.

\begin{definition}\emph{A mixed strategy $\boldsymbol{p}^*\textrm{$=$}(\boldsymbol{p}_1^*, \cdots, \boldsymbol{p}_J^*)\textrm{$=$}(\boldsymbol{p}_j^*, \boldsymbol{p}^*_{-j})$ constitutes a} mixed-strategy Nash equilibrium \emph{if, $\forall j \in \mathcal{J}$ and $\forall \boldsymbol{p}_j \in \Delta (\mathcal{A}_j)$, $\overline{u}_j(\boldsymbol{p}^*_j, \boldsymbol{p}^*_{-j})\geq \overline{u}_j(\boldsymbol{p}_j, \boldsymbol{p}^*_{-j})$.}
\end{definition}

Here, we note that any finite noncooperative game will admit at least one mixed-strategy Nash equilibrium~\cite{nash_existence}. To solve for the mixed-strategy NE of our game $\mathcal{G}$, we first consider the simpler scenario in which the number of SBSs is less than the number of unlicensed channels. Then, we develop a learning algorithm to handle the more realistic scenario in which the number of SBSs is much larger than the number of unlicensed channels.

\begin{remark}\label{proposition_disjoint_channels}
\emph{If the number of SBSs is less than the number of available unlicensed channels (i.e., $J\leq C$), then the mixed-strategy NE solution will simply reduce to a pure strategy that is reached when all SBSs occupy disjoint channels during each time epoch of the time window $T$.}
\end{remark}

To show this, we consider two cases depending on whether or not carrier aggregation is enabled. Let $M_c=1$. Consider the state in which each SBS is operating on a different unlicensed channel.
If SBS $j$ changes its channel from $c$ to $c'$ on which SBS $i$ is transmitting,
then it would have to share channel $c'$ with SBS $i$ in an equal weighted manner (based on the inter-operator fairness constraint).
This leads to a decrease in the reward function of SBS $i$ on channel $c'$ (and potentially for SBS $j$), which makes
SBS $i$ deviate to another channel that is less occupied (e.g., $c$). Therefore, a given strategy cannot be a best response (BR) strategy for SBS $i$ in case it results in its transmission on the same channel as SBS $j$. Therefore, all strategies that result in more than one SBS occupying the same channel are dominated by the alternative where
different SBSs transmit on disjoint channels and hence cannot correspond to BR strategies. Consequently, at the NE point, all SBSs play their BR strategies that would result in each SBS occupying a disjoint channel.
Similarly for $M_c>1$. If SBS $j$ transmits on multiple channels, then aggregating a channel that is already occupied by SBS $i$ would
make SBS $i$ change its operating channel to a less congested one. This implies that an SBS would not aggregate more channels unless
they are not occupied by other SBSs.

Therefore, we can conclude that our proposed scheme results in having less number of SBSs on each of the unlicensed bands. This leads to a lower collision probability on each channel and a better coexistence with WLAN. Moreover, enabling carrier aggregation does not necessarily allow LTE to offload more traffic to the unlicensed band. On the other hand, our proposed scheme can avoid causing performance degradation to WLAN in case a large number
of LTE operators deploy LTE-LAA in the unlicensed bands.

Now, when $J>C$, multiple SBSs will then potentially have to share the same channel. In this case, the mixed-strategy NE is challenging to characterize, and therefore, next, we propose a learning-based approach for solving our game $\mathcal{G}$. Given the fact that each SBS needs to learn a $\emph{sequence}$ of actions over the time window $T$ at $t=0$ based on a $\emph{sequence}$ of previous network states, the proposed learning algorithm must be capable of generating data that is sequential in nature. This necessitates the knowledge of historical traffic values as well as future network states for all the time epochs of the following time window $T$. Moreover, in order to satisfy the long-term fairness constraints (\ref{cons_6}) and (\ref{cons_7}), future actions cannot be assumed to be independent due to the long-term temporal dependence among these actions. Conventional reinforcement learning algorithms such as Q-learning and multi-armed bandit take as an input the current state of the network and enable the prediction of the next state only and therefore do not account for the interdependence of future actions~\cite{RL_intro}. To learn several steps ahead in time, recursive learning can be adopted. However, such an approach uses values already predicted, instead of measured past values which produces an accumulation of errors that may grow very fast. In contrast, deep learning techniques, such as time series prediction algorithms, are capable of learning long-term temporal dependence sequences based on input sequences~\cite{LSTM, RNN_survey}. This is viable due to their adaptive memory that allows them to store necessary previous state information to predict future events. Therefore, next, we develop a novel time series prediction algorithm based on deep learning techniques for solving the mixed-strategy NE of our game.






\section{RL-LSTM for Self-organizing Resource Allocation}\label{learning}
The proposed game requires each SBS to learn a sequence of actions over the prediction time window $T$, at $t=0$, without any knowledge of future network states. This necessitates a learning approach with memory for storing previous states whenever needed while being able to learn a sequence of future network states. Employing LSTMs is therefore an obvious choice for learning as they are capable of generating data that is sequential in nature~\cite{LSTM, RNN_survey, aidin_ICC}. Consequently, we propose a novel sequence level training algorithm based on RL-LSTM that allows SBSs to learn a sequence of future actions at operation time based on a sequence of historic traffic load thus maximizing the sum of their future rewards.


LSTMs are a special kind of ``deep" recurrent neural networks (RNNs) capable of storing information for long periods of time to learn the long-term dependency within a sequence~\cite{zaremba2014recurrent}. LSTMs process a variable-length sequence $\boldsymbol{y} = (y_1, y_2, \emph{...}, y_m)$ by incrementally adding new content into a single memory slot, with gates controlling the extent to which new content should be memorized, old content should be erased, and current content should be exposed. Unlike conventional one-step RL techniques (e.g., Q-learning), LSTM networks are capable of predicting a sequence of future actions~\cite{LSTM, aidin_magazine}. Predictions at a given time step are influenced by the network activations at previous time steps thus making LSTMs suitable for our application. The total number of parameters $W$ in a standard LSTM network with one cell in each memory block is given by:
\vspace{-0.1cm}
\begin{align}
W=n_c \times n_c \times 4 + n_i \times n_c \times4 + n_c \times n_o  + n_c \times 3
\end{align}
\noindent where $n_c$ is the number of memory cells, $n_i$ is the number of input units, and $n_o$ is the
number of output units. The computational complexity of learning
LSTM models per weight and time step is linear i.e., $O(1)$. Therefore,
the learning computational complexity per time step is $O(W)$~\cite{computational_complexity}.

Consequently, we consider an end-to-end RL-LSTM based approach to train the network to find a mixed-strategy NE of the game $\mathcal{G}$. LSTMs have three types of layers, one input and one output layer as well as a varying number of hidden layers depending on the dataset under study.
For our dataset, adding more hidden layers does not improve performance and thus one layer is sufficient. Moreover, in order to allow a sequence to sequence mapping, we consider an encoder-decoder model. The encoder network takes the input sequence and maps it to a vector of a fixed dimensionality. The encoded representation is then used by the decoder network to decode the target sequence from the vector. Fig.~\ref{model_1} summarizes the proposed approach.
The traffic encoder takes as an input the historical traffic loads and learns a vector representation of the input time-series. The multi-layer perceptron (MLP) summarizes the input vectors into one vector. In our scheme, an MLP is required to encode all the vectors together since a particular action at time $t$ depends on the values of all other input vectors (i.e., traffic values of all SBSs and WLAN on all the unlicensed channels). The action decoder takes as an input the summarized vector to reconstruct the predicted action sequence. All SBSs have the same input vector for the traffic encoders and thus they share the same traffic encoders. On the other hand, SBSs learn different action sequences and thus different SBSs use different action decoders.

In the first step, we need to train the neural networks in order to learn the parameters of the algorithm that would maximize the proposed utility function. Therefore, the proposed algorithm is divided into \emph{two phases, the training phase followed by the testing phase}.
In the former, SBSs are trained offline before they become active in the network using the architecture given in Fig.~\ref{model_1}. The input dataset represents the WiFi traffic load distribution on the unlicensed channels as well as the SBSs traffic load collected over several days. On the other hand, the testing phase corresponds to the actual execution of the algorithm after which the parameters have been optimized and is implemented on each SBS for execution during run time. 

\begin{figure}[t!]
  \begin{center}
  \vspace{-1.4cm}
  \centering
   \includegraphics[width=10cm]{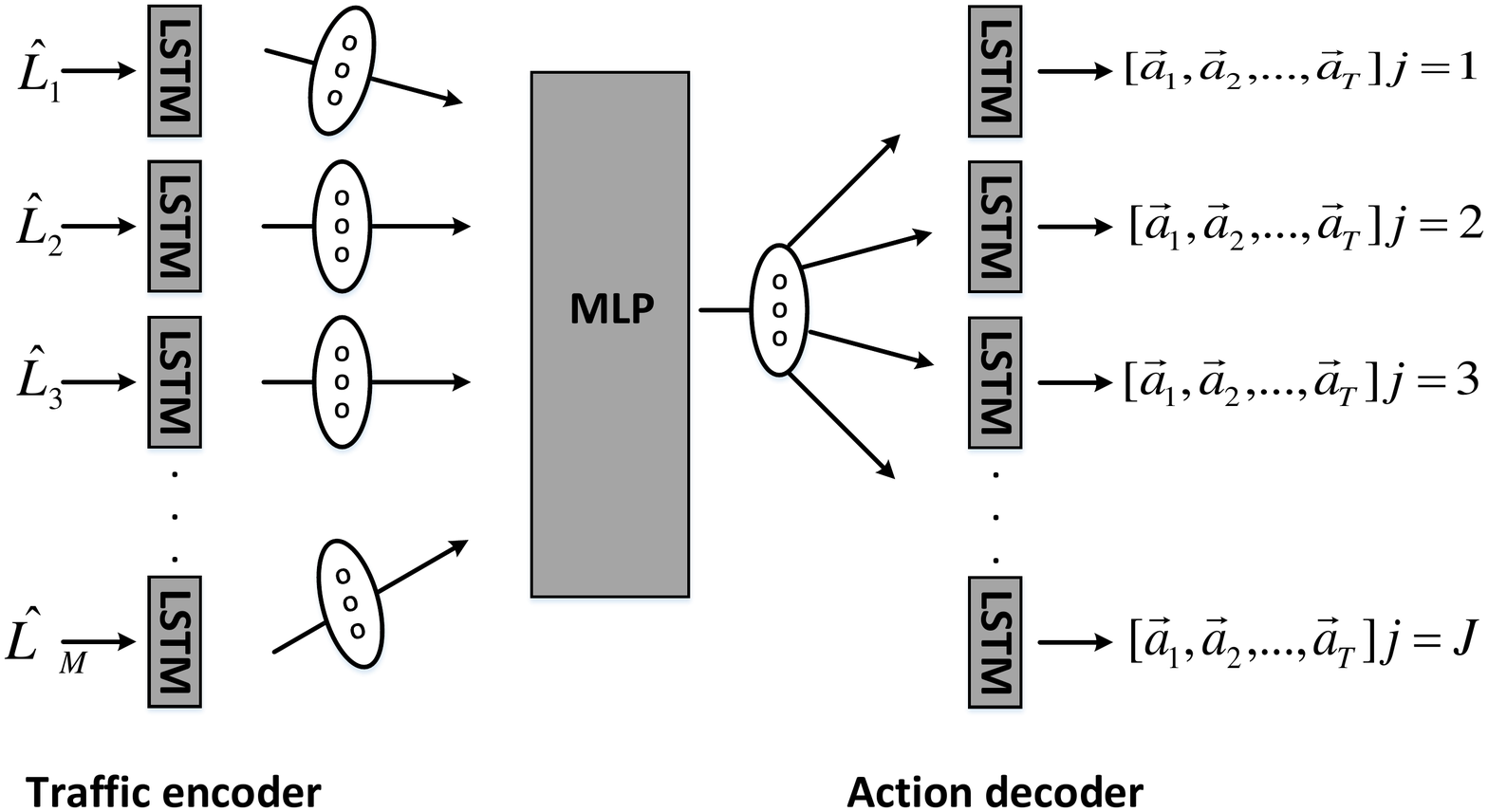} 
   \vspace{-0.6cm}
   \caption{Proposed framework.}\label{model_1}
      \vspace{-1.3cm}
  \end{center}
\end{figure}

For the training phase, we train the weights of our neural network using a policy gradient approach that aims at maximizing the expected return of a policy. This is achieved by representing the policy by its own function approximator and updating it according to the gradient of the expected reward with respect to the policy parameters~\cite{policy_gradient}.
Consider the set $\mathcal{M}$ of $M$ history traffic sequences corresponding to either an SBS or WiFi on each unlicensed channel, where $M=J+C$.
Let ${\boldsymbol{h}}_{m,t} \in \mathbb{R}^{n}$ and ${\boldsymbol{h}}_{j,t} \in \mathbb{R}^{n}$ be, respectively, the hidden vectors of the traffic encoder $m$ and action decoder of SBS $j$ at time~$t$. ${\boldsymbol{h}}_{m,t}$ and ${\boldsymbol{h}}_{j,t}$ are then computed by:
\vspace{0.1cm}
\begin{align}
{\boldsymbol{h} }_{m,t} \textrm{=} \phi \left( {\boldsymbol{v}}_{m,t}, {\boldsymbol{h}}_{m,t-1} \right), \;\; {\boldsymbol{h} }_{j,t} \textrm{=} \phi \left( {\boldsymbol{v}}_{j,t}, {\boldsymbol{h}}_{j,t-1} \right),
\end{align}
where $\phi$ refers to the LSTM cell
function~\cite{zaremba2014recurrent} being used, and
${\boldsymbol{v}}_{m,t}$ is the input vector. For the encoder,
${\boldsymbol{v}}_{m,t} = \left[\widehat{L}_{m,t}\right]$ is the history
traffic value. For the decoder,
${\boldsymbol{v}}_{j,t} = \left[{\boldsymbol{W}_d} \boldsymbol{e}(\boldsymbol{x}_{j,t-1})||\alpha_{j,c,t-1} \right]$
is the  vector of the previous predicted action where $\boldsymbol{e}()$
maps discrete value to a one-hot vector,
$\boldsymbol{W}_d \in \mathbb{R}^{n \times N_{x}}$ is a matrix that
is used to transform the discrete actions of each of the unlicensed channels into a vector,
and $N_{x}$ is the number of discrete actions.
In our approach, we learn the channel selection vector for all the channels simultaneously and
thus $\boldsymbol{x}_{j,t}=[x_{j,1,t}, \cdots, x_{j,C,t}]$.



To learn the mixed strategy of our proposed game, we need to initialize the action space with a subset of the continuous action space of $\alpha_{j,c,t}$.
A naive approach for working with continuous action spaces is to discretize the action space; however,
this approach would lead to combinatorial explosion and thus the well known problem of ``curse
of dimensionality" when highly discretizing our space and a loss in the accuracy of the predicted
action when considering less discretized values.
Therefore, we consider a sampling-based approach where we first define a probability distribution for the continuous variable $\alpha_{j,c,t}$ and for the discrete variable $x_{j,c,t}$
in order to deal with the large discrete action space as $T$ increases. We use a softmax classifier to predict the distribution for the discrete variable
$\boldsymbol{x}_{j,t}$ and a Gaussian policy for the distribution of the continuous variable $\alpha_{j,c,t}$. For the Gaussian policy, the probability of an action is
proportional to a Gaussian distribution with a parameterized mean and a fixed value for the variance in our implementation.
The variance of the Gaussian distribution defines the area around the mean from which we explore the action space.
For our implementation, the initial value of the variance is set to $0.06$ in order to increase exploration and then is decreased linearly towards $0.02$. Therefore, defining probability distributions for our variables allows the initialization of the action space $\mathcal{A}_j$ by sampling $Z$ actions from the proposed distributions. This enables the SBSs to
learn more accurate transmission probabilities for $\alpha_{j,c,t}$, as opposed to fixed discretization, thus satisfying the fairness constraints.
The hidden vector ${\boldsymbol{h}}_{j,t}$ in the decoder is used to predict the \mbox{$t$-th}~output actions $\boldsymbol{x}_{j,t}$ and $\alpha_{j,c,t}$. The probability vector over $\boldsymbol{x}_{j,t}$ and $\alpha_{j,c,t}$ can be defined, respectively, as:
\vspace{-0.7cm}
\begin{align}
\boldsymbol{x}_{j,t}| \boldsymbol{x}_{j,<t} , \alpha_{j,c,<t} , \widehat{\boldsymbol{L}}_{t} \sim {\sigma \left( \boldsymbol{W}_x {\boldsymbol{h}}_{j,t} \right)},
\end{align}
\vspace{-0.8cm}
\begin{align}
\mu_{j,c,t} = {S \left( \boldsymbol{W}_\mu {\boldsymbol{h}}_{j,t} \right)}, \;\;\alpha_{j,c,t}\sim \mathcal{N}(\mu_{j,c,t}, \mathrm{Var}(\alpha_{j,c,t})),
\end{align}

\noindent where $\mu_{j,c,t}$ and $\mathrm{Var}(\alpha_{j,c,t})$ correspond to the mean value and variance of the Gaussian policy respectively, $\boldsymbol{W}_x \in \mathbb{R}^{|V_a| \times n}, \boldsymbol{W}_\mu \in \mathbb{R}^{n}$ are parameters, $\sigma(.)$ is the softmax function $\sigma (\boldsymbol{b})_q =\frac{e^{b_q}}{\sum_{o=1}^O e^{b_o}}$ for $q=1, \cdots, O$, and $S(.)$ is the sigmoid function where $S(b) = \frac{1}{1+e^{-b}}$ and is used to normalize the value to $(0,1)$. $\alpha_{j,c,t}$ is computed only when $x_{j,c,t}=1$.
The probability of the whole action sequence for SBS $j$, given a historic traffic sequence $\widehat{\boldsymbol{L}}$, $p_{j,\boldsymbol{a}_j|\widehat{\boldsymbol{L}}}$, is given by:
\begin{align}
p_{j,\boldsymbol{a}_j|\widehat{\boldsymbol{L}}} = \prod_{t = 1}^{T}{ p\left( (\boldsymbol{x}_{j,t}, \alpha_{j,c,t})| \boldsymbol{x}_{j,<t}, \alpha_{j,c,<t}, \widehat{\boldsymbol{L}}_{t}\right)},
\end{align}
where $\widehat{\boldsymbol{L}}_{t}\textrm{$=$}(\widehat{L}_{1,t}, \cdots, \widehat{L}_{M,t})$, $\boldsymbol{x}_{j,<t} \textrm{$=$} [\boldsymbol{x}_{j,1}, \cdots, \boldsymbol{x}_{j,t-1}]$, and $\mu_{j,c,<t} \textrm{$=$} [\mu_{j,c,1}, \cdots, \mu_{j,c,t-1}]$.

Our goal is to maximize the exact expectation of the reward $\widehat{u}_{j}(\boldsymbol{a}_{j},\boldsymbol{a}_{-j})$ over the action space for the training dataset. Therefore, the objective function can be defined as:
\begin{align}\label{objective}
\max_{\mathbf{\boldsymbol{\emph{a}}_\emph{j} \in \mathcal{A}_\emph{j}}} \sum_{\mathcal{D}} \overline{u}_j(\boldsymbol{p}_j, \boldsymbol{p}_{-j}),
\end{align}
where $\mathcal{D}$ is the training dataset. For this objective function, the REINFORCE algorithm~\cite{reinforce} can be used to compute the gradient of the expected reward with respect to the policy parameters, and then standard gradient descent optimization algorithms~\cite{policy_gradient} can be adopted to allow the model to generate optimal action sequences for input history traffic values. Specifically, Monte Carlo sampling is adopted to compute the expectation.

In particular, we adopt the RMSprop gradient descent optimization algorithm for the update rule~\cite{rmsprop}. The learning rate of a particular weight is divided by a running average of the magnitudes of recent gradients for that weight. The RMSprop update rule is given by:
\begin{align}
\mathds{E}[g^2]_{t}=\gamma \mathds{E}[g^2]_{t-1}+ (1-\gamma) g_t^2,
\end{align}
\vspace{-0.9cm}
\begin{align}
\theta_{t+1}=\theta_t - \frac{\lambda}{\sqrt{\mathds{E}[g^2]_{t}+\epsilon}}g_t,
\end{align}
where $\theta_t$ corresponds to the model parameters at time $t$, $g_{t}$ is the gradient of the objective function with respect to the parameter $\theta$ at time step $t$, $\mathds{E}[g^2]_{t}$ is the expected value of the magnitudes of recent gradients, $\gamma$ is the discount factor, $\lambda$ is the learning rate and $\epsilon$ is a smoothing parameter.

\begin{algorithm}[t!] \scriptsize
\caption{Training phase of the proposed approach.}
\label{Training_RL_LSTM_algorithm}
\begin{algorithmic}[t!]
\STATE \textbf{Input}: $~\mathcal{J}; \mathcal{W}; \mathcal{C}; \widehat{L}_{j,t} \forall j \in \mathcal{J}, t; \widehat{L}_{w,c,t} \forall c\in \mathcal{C}~, t$.
\STATE \textbf{\emph{Initialization}:} The weights of all LSTMs are initialized following a uniform distribution with arbitrarily small values.
\STATE \textbf{\emph{Training}:} Each SBS $j$ is modeled as an LSTM network.
\WHILE{Any of the coupled constraints is not satisfied}
\FOR {Number of training epochs}
\FOR{Size of the training dataset}



\STATE \textbf{Step 1.} Run Algorithm \ref{Testing_RL_LSTM_algorithm} to compute the best actions for all SBSs.
\FOR{$j$=1:$J$}
\STATE \textbf{Step 2.} Sample actions for SBS $j$ based on the best expected actions of other SBSs. 
\STATE \textbf{Step 3.} Use REINFORCE~\cite{reinforce} to update rule and compute the gradient of the expected value of the reward function.
\STATE \textbf{Step 4.} Update model parameters with back-propagation algorithm~\cite{BPalgorithm}.

\ENDFOR
\ENDFOR
\ENDFOR

\textbf{Step 5.} Using the incremental penalty algorithm, check the feasibility of the coupled constraints and update the values of $\boldsymbol{\rho}_l$ accordingly.
\ENDWHILE
\end{algorithmic}
\end{algorithm}

\begin{algorithm}[t!] \scriptsize
\caption{Testing phase of the proposed approach.}
\label{Testing_RL_LSTM_algorithm}

\begin{algorithmic}[t!]
\STATE \textbf{Input:} $~\mathcal{J}; \mathcal{W}; \mathcal{C}; \widehat{L}_{j,t} \forall j \in \mathcal{J}, t; \widehat{L}_{w,c,t} \forall c\in \mathcal{C}~, t$.

\FOR{For each SBS $j$}
\STATE \textbf{Step 1. \emph{Traffic history encoding}:}~
The history traffic of each SBS and WLAN activity on each channel is fed into each of the $M$ LSTM traffic encoders.

\STATE \textbf{Step 2. \emph{Vector summarization}:}~The encoded vectors are transformed to initialize action decoders.

\STATE \textbf{Step 3. \emph{Action decoding}:}~Action sequence is decoded for each SBS $j$.

\ENDFOR
\end{algorithmic}
\end{algorithm}

On the other hand, the testing phase corresponds to the actual execution of the algorithm on each SBS. Based on historical traffic values, each SBS learns the future sequence of actions based on the learned parameters from the training phase.
For practicality, we assume knowledge of historical measurements of the WiFi activity on each of the unlicensed channels using simple network management protocol (SNMP) statistics with accurate calibration~\cite{icnp_07}
and of other SBSs by exchanging past traffic information via the X2 interface as done in~\cite{mswim} and~\cite{CU_LTE}. For our proposed scheme, the SBSs are trained over a large training dataset taking into account the traffic load over multiple days. The likeliness that an error occurs at the same time over multiple days is thus very rare. Moreover, our proposed scheme takes into account a \emph{sequence} of history traffic values. Therefore, in case of non-ideal information, the impact of this error can be considered to be negligible. The proposed approach can also be combined with online machine learning~\cite{online_learning} to accommodate changes in the traffic model, by properly re-training the developed learning mechanism. Consequently, the proposed algorithm offers a practical solution that is amenable to implementation. Here, we note that one practical challenge for deploying this algorithm in a real-world network is synchronization between SBSs and WAPs. In essence, such synchronization can be achieved by inter-operator cooperation, using mechanisms such as in~\cite{omid_sync}. The training and the testing phases are given in Algorithms~\ref{Training_RL_LSTM_algorithm} and~\ref{Testing_RL_LSTM_algorithm} respectively.

Note that guaranteeing the convergence of the proposed algorithm is challenging as it is highly dependent on the hyperparameters used during the training phase. It has been shown in~\cite{learning_rate} that the learning rate and the hidden layer size
are the two most important hyperparameters for the convergence of LSTMs. For instance, using too few neurons in the hidden layers results in underfitting which could
make it hard for the neural network to detect the signals in a complicated data set.
On the other hand, using too many neurons in the hidden layers can result in either overfitting~\cite{overfitting} or an increase in the training time. Therefore, in this work, we limit our contribution to providing simulation results (see Section~\ref{simulation_section}) to show that, under a reasonable choice of the hyperparameters, convergence is observed for our proposed game, as per the following theorem:

\begin{theorem}~\label{theorem_NNE} \emph{If Algorithm~\ref{Training_RL_LSTM_algorithm} converges, then the convergence strategy profile corresponds to a mixed-strategy NE of game $\mathcal{G}$.}
\end{theorem}
\begin{proof}
In order to prove this theorem, we first need to show that the solution of the adopted multi-agent learning algorithm converges to an equilibrium point. In fact, every strict Nash equilibrium is a local optimum for a gradient descent learning approach but the reverse is not always true (Theorems 2 and 3 in~\cite{IMP_NE_gradient_descent}). Therefore, to show that a gradient-based learning method guarantees convergence of our proposed game to an equilibrium point, we define the following lemma.

\begin{lemma}~\label{lemma_1}
The square of a linear function is convex. It follows that the payoff function of player $j$ defined in (\ref{penalized_utility}) is an affine combination of convex functions, and hence is convex. Therefore, a gradient-based learning algorithm for our game $\mathcal{G}$ allows the convergence to an equilibrium point of that game.
\end{lemma}
Lemma~\ref{lemma_1} is the consequence of the convexity of the players' payoffs where it has been shown in~\cite{NE_gradient} that under certain convexity assumptions about the shape of payoff functions,
the gradient-descent process converges to an equilibrium point.
However, convergence is only guaranteed under a decreasing step-size sequence~\cite{robbins_monro}. Therefore, given the fact that we employ an adaptive learning rate method satisfying the Robbins-Monro conditions $(\lambda >0, \sum_{t=0}^\infty \lambda(t)=+\infty, \sum_{t=0}^\infty \lambda^2(t)<+\infty)$, one can guarantee that under suitable initial conditions, our proposed algorithm converges to an equilibrium point.





Moreover, following the penalized reformulation of our game $\mathcal{G}$, one can easily show that a strategy that violates the coupled constraints cannot be a best response strategy. From~\cite{fukushima}, there exists $\boldsymbol{\rho}_l^*$ such that the incremental penalty algorithm terminates. Therefore, there exists a mixed strategy for which the coupled constraints are satisfied at $\boldsymbol{\rho}_l^*$. In that case, there is no incentive for an SBS to violate any of the coupled constraints, otherwise, its reward function would be penalized by the corresponding penalty function. Hence, all strategies that violate the coupled constraints are dominated by the alternative of complying with these constraints. Since in the proposed algorithm, the optimal strategy profile results in maximizing ${\mathds{E}_{\boldsymbol{p}_{j}} \left[\widehat{u}_{j}\left( \boldsymbol{a}_{j},\boldsymbol{a}_{-j} \right)\right]}$, we can conclude that the converged mixed-strategy NE is guaranteed not to violate the coupled constraints and hence it corresponds to a mixed-strategy NE for the game $\mathcal{G}$.
Therefore, our proposed learning algorithm learns a mixed strategy of the game $\mathcal{G}$, by using a deep neural network function approximator to represent strategies, and by averaging those strategies via gradient descent machine learning techniques.





\end{proof}

\vspace{-0.2cm}
\begin{table}[t!]\footnotesize
\setlength{\belowcaptionskip}{0pt}
\setlength{\abovedisplayskip}{3pt}
\captionsetup{belowskip=0pt}
\newcommand{\tabincell}[2]{\begin{tabular}{@{}#1@{}}#1.6\end{tabular}}
 \setlength{\abovecaptionskip}{2pt}
 \renewcommand{\captionlabelfont}{\small}
\caption[table]{\scriptsize{\\SYSTEM PARAMETERS}}\label{parameters}
\centering
\begin{tabular}{|c|c|c|c|}
\hline
\textbf{Parameters} & \textbf{Values} & \textbf{Parameters} & \textbf{Values} \\
\hline
Transmit power ($P_t$) & 20 dBm & BW (channel) & 20 MHz \\
\hline
CCA threshold & -80 dBm & Noise variance & 92 dBm/Hz \\
\hline
Path loss & $15.3+50\log_{10}(m)$ & SIFS & 16 $\mu$s \\
\hline
Hidden size (encoder) & 70 & DIFS & 34 $\mu$s\\
\hline
Hidden size (decoder) & 70 & $\mathrm{CW_{min}}$ & 15 slots\\
\hline
time epoch ($t$) & 5 min & $\mathrm{CW_{max}}$ & 1023 slots\\
\hline
Action sampling ($Z$) & 100 samples & ACK & 256 bits\\
\hline
History traffic size & 7 time epochs &  $P_{\mathrm{LTE}},P_{\mathrm{WiFi}}$ & 1, 1\\
\hline
Learning rate ($\lambda$) & 0.01 & LSTM layers & 1\\
\hline
Learning rate decay ($\gamma$) & 0.95 & $t_{\mathrm{max}}$ & 0.9\\
\hline
\end{tabular}
\vspace{-0.5cm}
\end{table}

\vspace{-0.8cm}
\section{Simulation Results and Analysis}\label{simulation_section}
For our simulations, we consider a $300$~m $\times$ $300$~m square area in which we randomly deploy a number of SBSs and 7 WAPs that share $7$ unlicensed channels. We use real data for traffic loads from the dataset provided in~\cite{IBM_dataset} and divide it as $80$\% for training and $20$\% for testing. During the training phase, we randomly shuffle examples in the training dataset in order to prevent cycles when approximating the reward function. Table~\ref{parameters} summarizes the main simulation parameters. All statistical results are averaged over a large number of independent runs.

\begin{figure}[t!]
  \begin{center}
  \vspace{-0.2cm}
    \includegraphics[width=11cm]{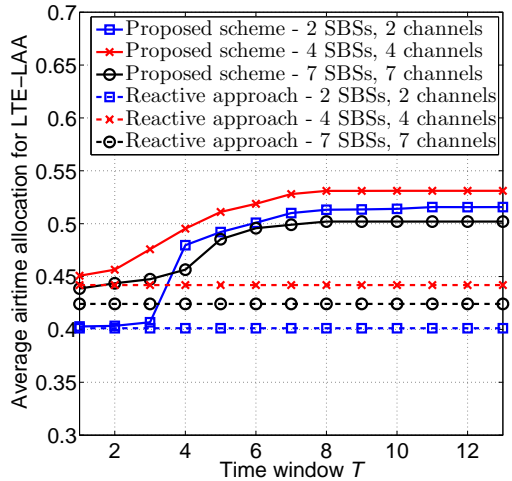}
   \caption{The average throughput gain for LTE-LAA upon applying a proactive approach (with varying $T$) as compared to a reactive approach.}\label{proactive_reactive} 
   \vspace{-1.2cm}
  \end{center}
\end{figure}

Fig.~\ref{proactive_reactive} shows the average throughput gain, compared to a reactive approach, achieved by the proposed approach for different values of $T$
under three different network scenarios. Here, we note that, in Fig.~\ref{proactive_reactive}, the case in which $T=1$ corresponds to other proactive schemes such as exponential smoothing and conventional reinforcement learning algorithms (e.g., Q-learning and multi-armed bandit)~\cite{RL_intro}. Intuitively, a larger ˜$T$ provides the framework additional opportunities to benefit over the reactive approach, which does not account for future traffic loads. First, evidently, for very small time windows, our proposed approach does not yield any significant gains.
However, as $T$ increases, LTE-LAA network utilizes statistical predictions for allocating resources and thus the gains start to become more pronounced as compared to the reactive approach as well as to other proactive approaches at $T=1$. For example, from Fig.~\ref{proactive_reactive}, we can see that, for $4$ SBSs and $4$ channels, our proposed scheme achieves an increase of 17\% and 20\% in the average airtime allocation for LTE-LAA as compared to other proactive schemes and the reactive approach, respectively. Eventually, as ˜$T$ grows, the gain of our proposed framework remains almost constant at the maximum achievable value. This corresponds to the minimum value of $T$ required to allow the LTE-LAA network smooth out its load over time and thus achieve maximum gain while guaranteeing fairness to WLAN.



\begin{figure}[t!]
  \begin{center}
    \includegraphics[width=11cm]{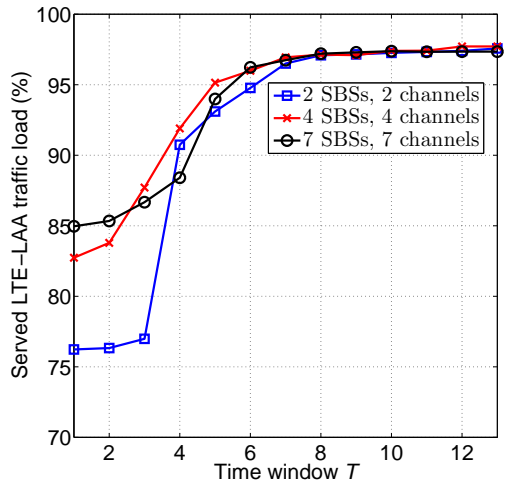}
    \vspace{-0.1cm}
   \caption{The proportion of load served over LTE-LAA as a function of $T$.}\label{offered_served} 
     \vspace{-1cm}
  \end{center}
\end{figure}

Fig.~\ref{offered_served} shows the proportion of LTE-LAA served load for different values of $T$. Clearly, as $T$ increases, the proportion of LTE-LAA served traffic increases. For example, the proportion of served load increases from $82$\% to $97$\% for the case of $4$ SBSs and $4$ channels. The gain of the LTE-LAA network stems from the flexibility of choosing actions over a large time horizon $T$. In contrast to the myopic reactive approach, our proposed proactive scheme takes into account future predictions of the network state along with the current state. Therefore, the optimal policy will balance the instantaneous reward and the available information for future use and thus maximizing the total load served over time. 



\begin{figure}
\vspace{-0.5cm}
\begin{subfigure}{1.0\textwidth}
  \centering
  \includegraphics[width=12cm]{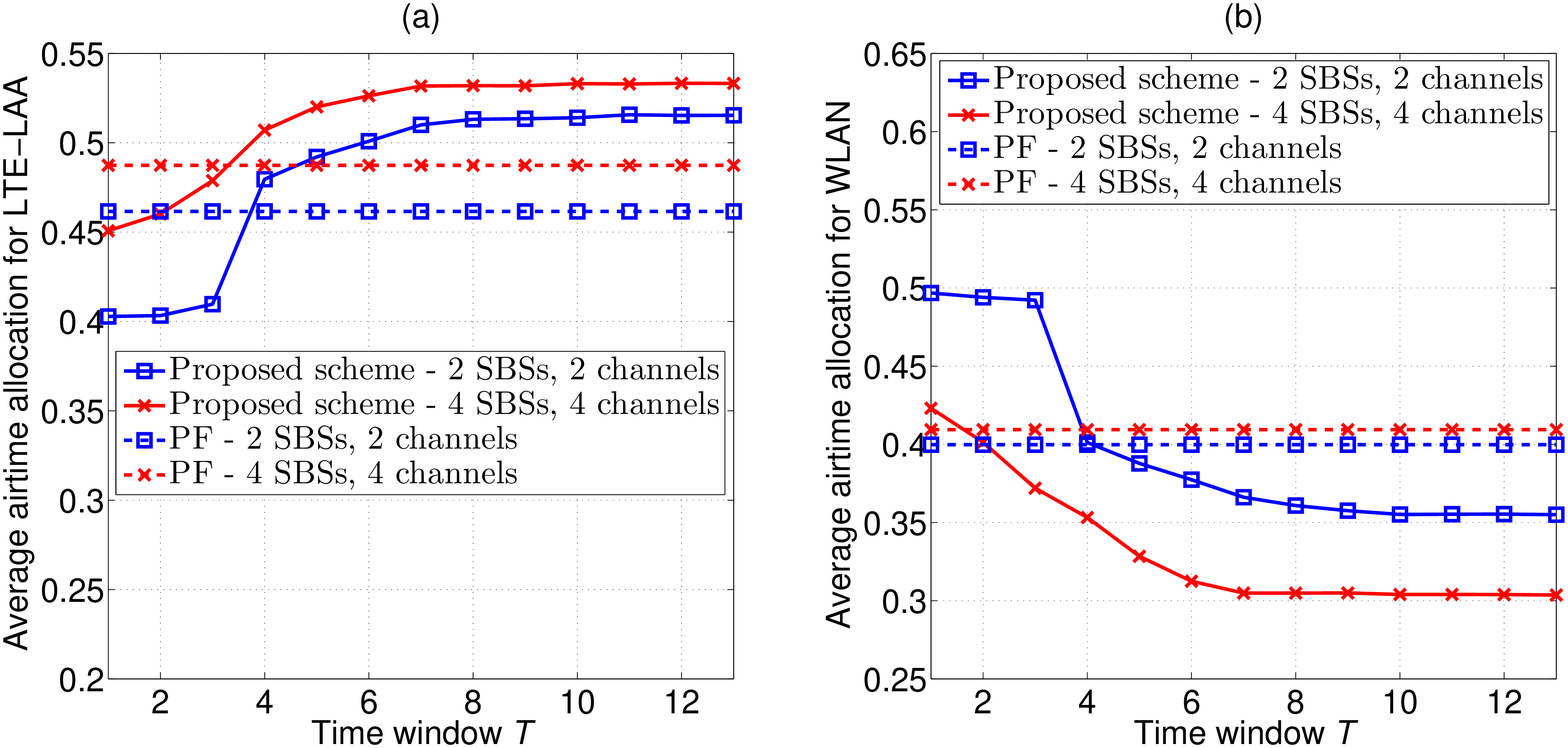}
  \label{fig:sfig1}
\end{subfigure}\\
\begin{subfigure}{1.0\textwidth}
  \centering
  \hspace{-0.15cm}\includegraphics[width=12cm]{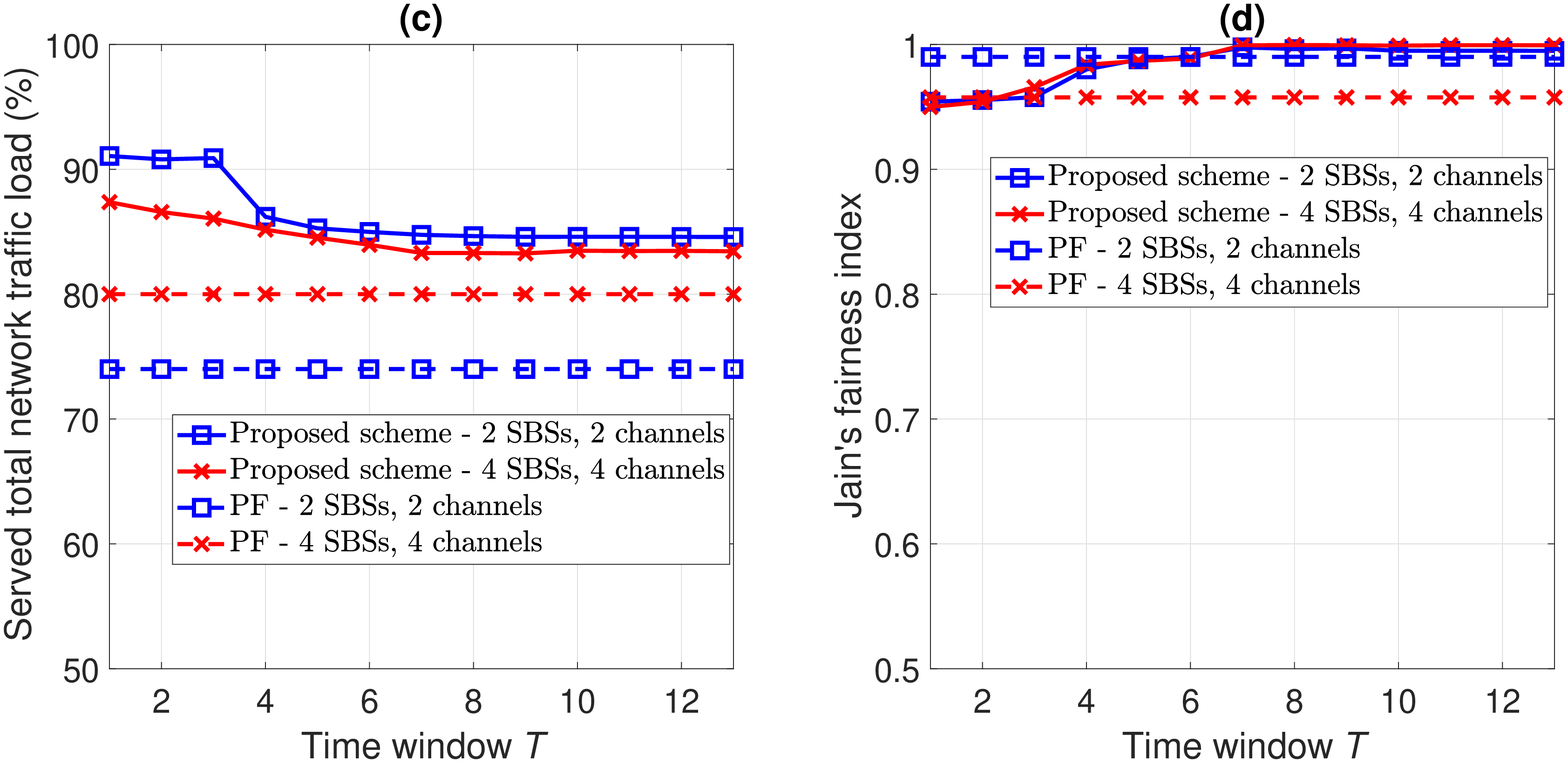}
  \label{fig:sfig2}
\end{subfigure}
\vspace{-0.1cm}
\caption{The (a) average airtime allocated for LTE-LAA, (b) average airtime allocated for WLAN, (c) proportion of served total network traffic load, and (d) Jain's fairness index resulting from our proposed scheme as well as from a centralized proportional fairness utility maximization scheme (with varying $T$).}
\label{PF_comparison}
\vspace{-0.8cm}
\end{figure}

Fig.~\ref{PF_comparison} shows the (a) average airtime allocated for the LTE-LAA network, (b) average airtime allocated for WLAN, (c) proportion of served total network traffic load, and (d) Jain's fairness index as a function of $T$ resulting from our proposed scheme as well as from a centralized solution considering a proportional
fairness (PF) utility function that is widely used for resource allocation~\cite{PF}, subject to constraints (\ref{cons_1})-(\ref{cons_5}) with $T=1$. Here, we compute the Jain's index based on the proportion of served traffic load for each network using $\mathcal{J}(l_o)=\frac{(\sum_{o=1}^O l_o)^2}{O \cdot \sum_{o=1}^O l_o^2}$, where $l_o$ is the proportion of served traffic load for network $o$ and $O$ is number of networks~\cite{jain}. The centralized solution of the PF resource allocation is obtained using the branch-and-bound algorithm in~\cite{bonmin}. From Fig.~\ref{PF_comparison} (a), we can see that for small
values of $T$, the PF allocation offers higher airtime allocation for the LTE-LAA network. For example, for the scenario of $4$ SBSs and $4$ channels, PF offers airtime gains of $7$\% and $5$\% as compared to our proposed approach for $T=1$ and $2$ respectively.
However, as $T$ increases, our proposed scheme achieves more transmission opportunities for the LTE-LAA network as compared to the PF solution. For instance, for the scenario of
$2$ SBSs and $2$ channels, our proposed scheme achieves an increase of $11$\% in the transmission opportunities for $T\geq 8$. This gain stems from the proactive resource allocation approach that allows more flexibility in spectrum allocation as $T$ increases. From Figs.~\ref{PF_comparison} (b) and (c), we can see that, although the average airtime allocation for WiFi resulting from our proposed scheme is less than that of the PF scheme for $T>4$, the proportion of the total network traffic load served by our proposed scheme is higher than that of the PF scheme for all values of $T$ (e.g., 84\% for our proposed scheme as compared to 74\% for PF for the case of 2 SBSs and 2 channels and for $T>6$). Moreover, from Fig.~\ref{PF_comparison} (d), we can conclude that, as $T$ increases, our proposed scheme achieves similar fairness performance as that of PF. This is due to the fact that, for our proposed scheme, as $T$ increases, the proportion of LTE served traffic load increases while that of WiFi decreases eventually, converging to a constant value for $T>7$. In particular, a relatively large time window allows SBSs to exploit future off-peak hours on the unlicensed band and thus increasing their transmission opportunities. Therefore, at the convergence point, the proportion of served traffic load of both technologies is almost the same. In summary, our proposed scheme allows more transmission opportunities for LTE-LAA, increases the proportion of the total network served load while also preserving fairness with WiFi. It offers better tradeoff in terms of efficiency and fairness compared to the centralized PF allocation scheme. Note that the resulting problem for the PF solution is a mixed integer nonlinear optimization problem (MINLP) and therefore, finding its solution becomes challenging for larger network scenarios due to the polynomial computational complexity.

\begin{figure}
\vspace{-0.4cm}
\begin{subfigure}{1.0\textwidth}
  \centering
  \includegraphics[width=12cm]{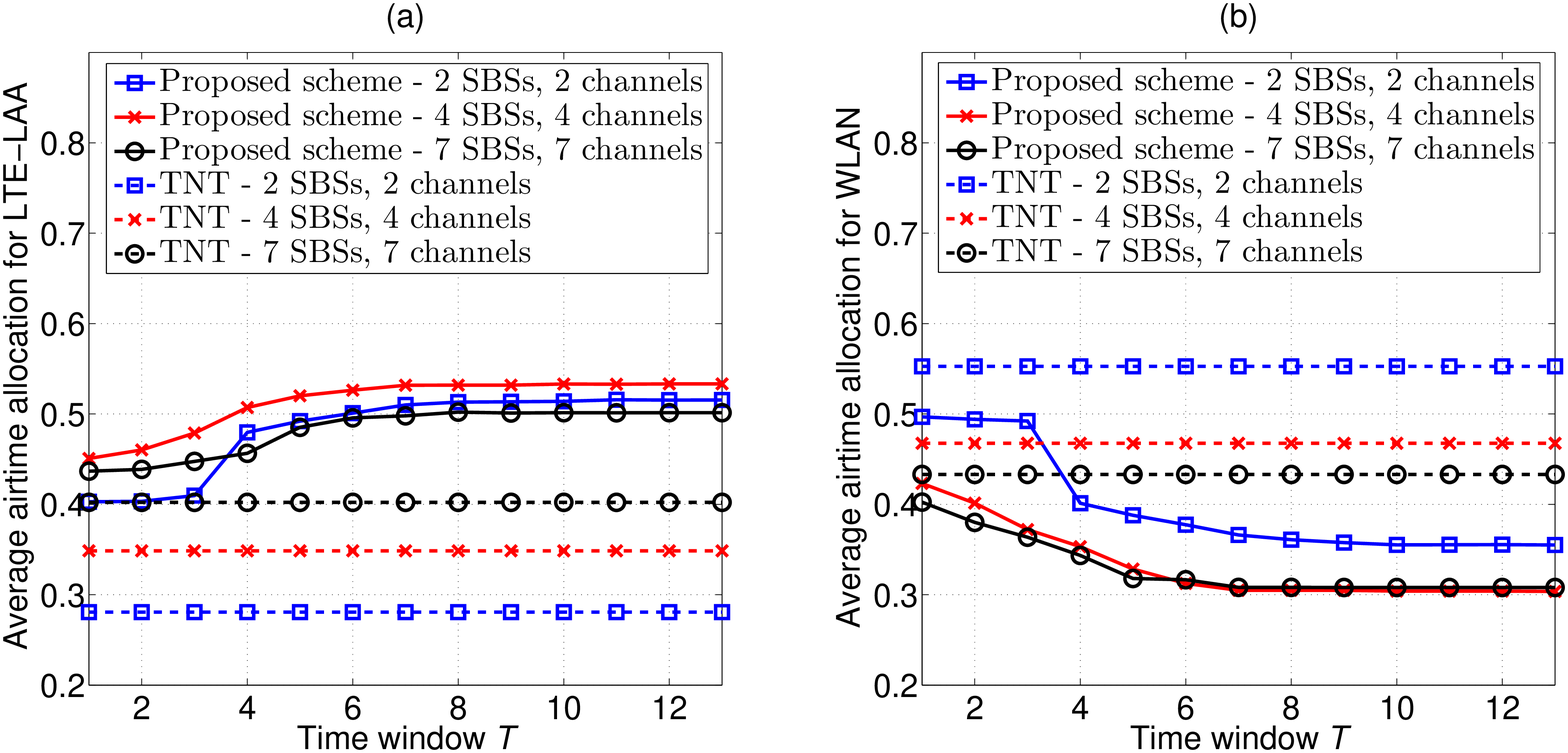}
  \label{fig:sfig1}
\end{subfigure}\\
\begin{subfigure}{1.0\textwidth}
  \centering
  \hspace{-0.25cm}\includegraphics[width=12cm]{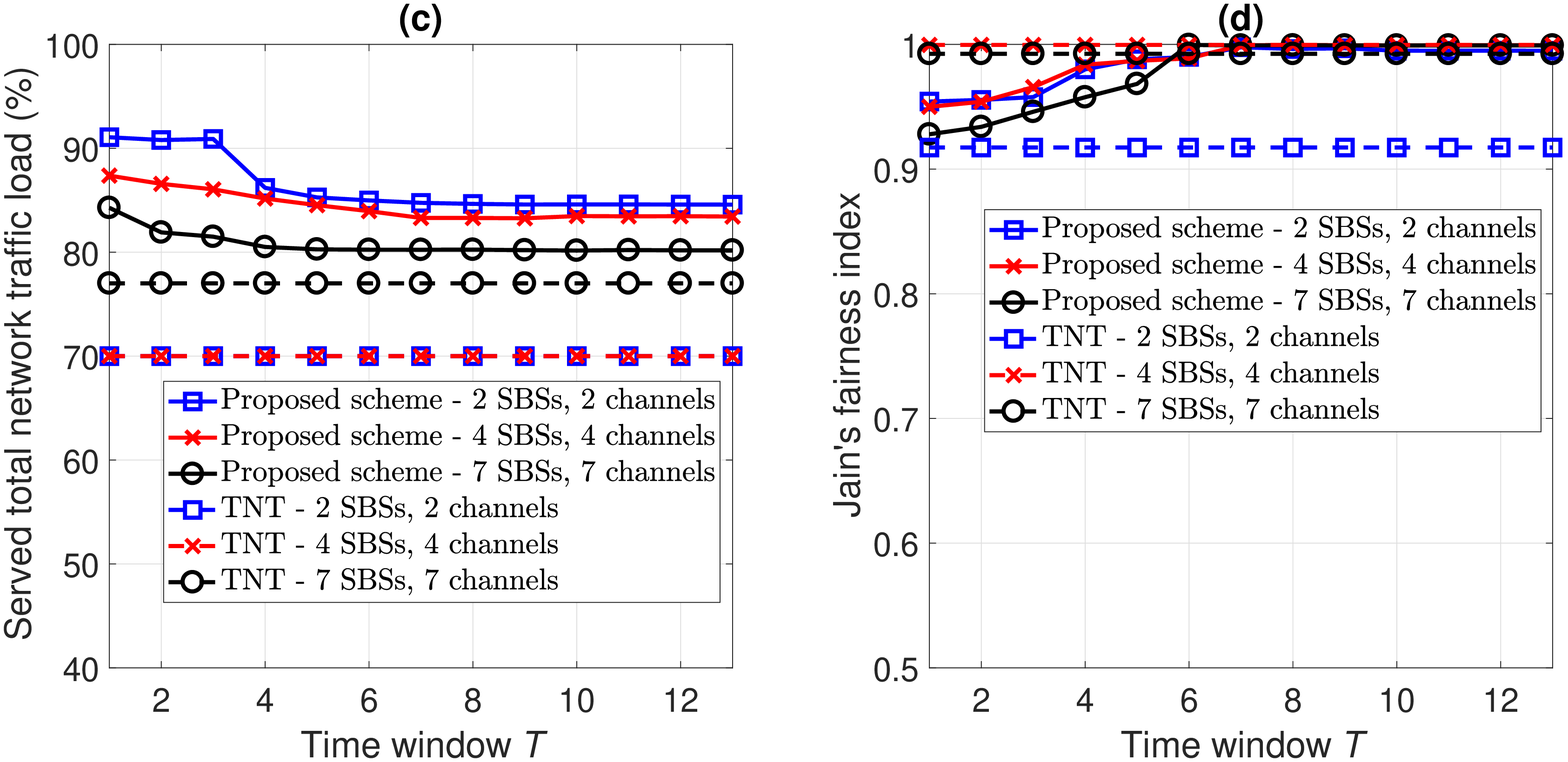}
  \label{fig:sfig2}
\end{subfigure}
\vspace{-0.1cm}
\caption{The (a) average airtime allocated for LTE-LAA, (b) average airtime allocated for WLAN, (c) proportion of served total network traffic load, and (d) Jain's fairness index resulting from our proposed scheme as well as from a centralized total network throughput utility maximization scheme (with varying $T$).}
\label{TNT_comparison}
\vspace{-0.7cm}
\end{figure}

Fig.~\ref{TNT_comparison} shows the (a) average airtime allocated for the LTE-LAA network, (b) average airtime allocated for WLAN, (c) proportion of served total network traffic load, and (d) Jain's fairness index as a function of $T$ resulting from our proposed scheme as well as a centralized solution considering a total network throughput (TNT) utility function subject to constraints (\ref{cons_1})-(\ref{cons_5}) with $T=1$. From Fig.~\ref{TNT_comparison} (a), we can see that our proposed resource allocation scheme offers higher transmission opportunities for LTE-LAA for all values of $T$ as compared to the centralized solution considering a TNT utility function. For example, for the case of $4$ SBSs and $4$ channels, the gain for our proposed approach can reach up to $52$\% for $T\geq 8$.
Similarly, from Figs.~\ref{TNT_comparison} (b) and (c), we can observe that, although the average airtime allocation for WLAN for our proposed scheme is less than that of the TNT scheme for all $T$, the proportion of the total served network traffic load for our proposed scheme is higher than that of the TNT scheme. From Fig.~\ref{TNT_comparison} (d), we can also conclude that, as $T$ increases, our scheme achieves similar fairness to TNT due to the fact that, as $T$ increases, the proportion of LTE served traffic load increases while that of WiFi decreases for our proposed scheme, converging to a constant value for $T>7$. At this convergence point, the proportion of served traffic load of both technologies is almost the same. In summary, our proposed scheme offers a better tradeoff in terms of efficiency and fairness as compared to the centralized TNT allocation scheme.


\begin{figure}[t!]
  \begin{center}
  \vspace{-0.7cm}
   \includegraphics[width=14.5cm]{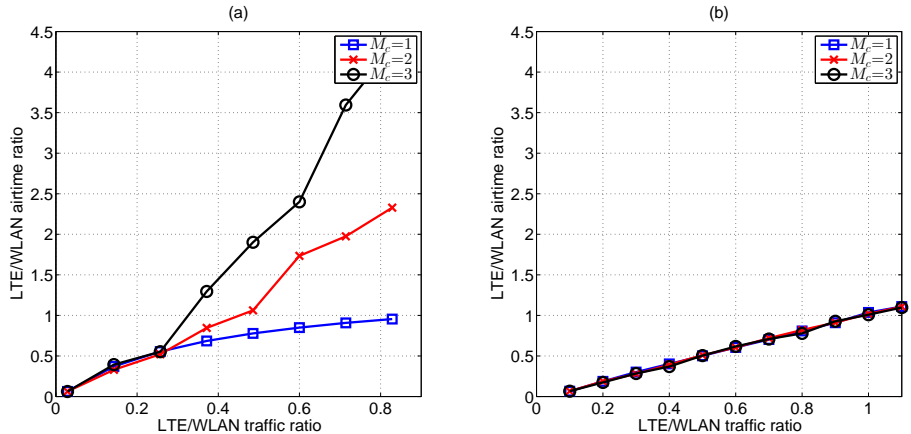}
    \vspace{-0.7cm}
   \caption{LTE/WLAN airtime ratio as a function of the LTE/WLAN
   traffic ratio for $3$ different values of $M_c$ ($M_c=1$, $2$ and $3$). The LTE and WLAN airtime fraction correspond to the average airtime allocated per SBS and per WAP, respectively. Moreover, the number of unlicensed channels is fixed to 7 and the number of SBSs is equal to $2$ and $7$ in (a) and (b) respectively.}\label{airtimeratio}
     \vspace{-1.2cm}
  \end{center}
\end{figure}

Fig.~\ref{airtimeratio} shows the value of the LTE/WLAN airtime ratio under varying LTE/WLAN traffic ratio and for different values of $M_c$. Note here that the LTE and WLAN airtime fraction correspond to the average airtime allocated per SBS and per WAP, respectively. We consider two different scenarios with varying number of SBSs ($2$ and $7$ SBSs for scenarios (a) and (b) respectively),
while the number of unlicensed channels is fixed to $7$. Fig.~\ref{airtimeratio} shows that inter-technology fairness is satisfied. This can be clearly seen in scenario (b) for the case of $M_c=1$. For instance, when the traffic ratio is $1$, LTE/WLAN airtime ratio is $1$ and thus equal weighted airtime is allocated for each technology (given that $P_{\mathrm{LTE}}=1$ and $P_{\mathrm{WiFi}}=1$).
From Fig.~\ref{airtimeratio}, we can also see that enabling carrier
aggregation impacts the resource allocation outcome. In fact, we can
see that a considerable gain in terms of spectrum access time can be
achieved with carrier aggregation. For instance, in the case of $2$ SBSs, the LTE/WLAN airtime ratio increases from $0.84$ for $M_c=1$ to $1.7$ and $2.4$ for $M_c=2$ and $3$ respectively for the value of $0.6$ for LTE/WLAN traffic ratio. On the other hand, this gain decreases as more SBSs are deployed and for a densely deployed LTE-LAA network, there is no need to aggregate more channels. This can be seen from (b) where the LTE-LAA network gets the same airtime share for $M_c=1$, $2$ and $3$ (as also shown in Remark~\ref{proposition_disjoint_channels}).

Moreover, Fig.~\ref{airtimeratio} shows that deploying more SBSs does not necessarily allow more airtime for the LTE-LAA network.
For example, LTE/WLAN airtime ratio of scenarios (a) and (b)
corresponding to $0.6$ LTE/WLAN traffic ratio is equal to $0.84$ and $0.6$ respectively for $M_c=1$. Consequently, the proposed scheme can avoid causing performance degradation to WLAN in the case LTE operators selfishly deploy a high number of SBSs.

\begin{figure}[t!]
  \begin{center}
  \vspace{-0.2cm}
    \includegraphics[width=13cm]{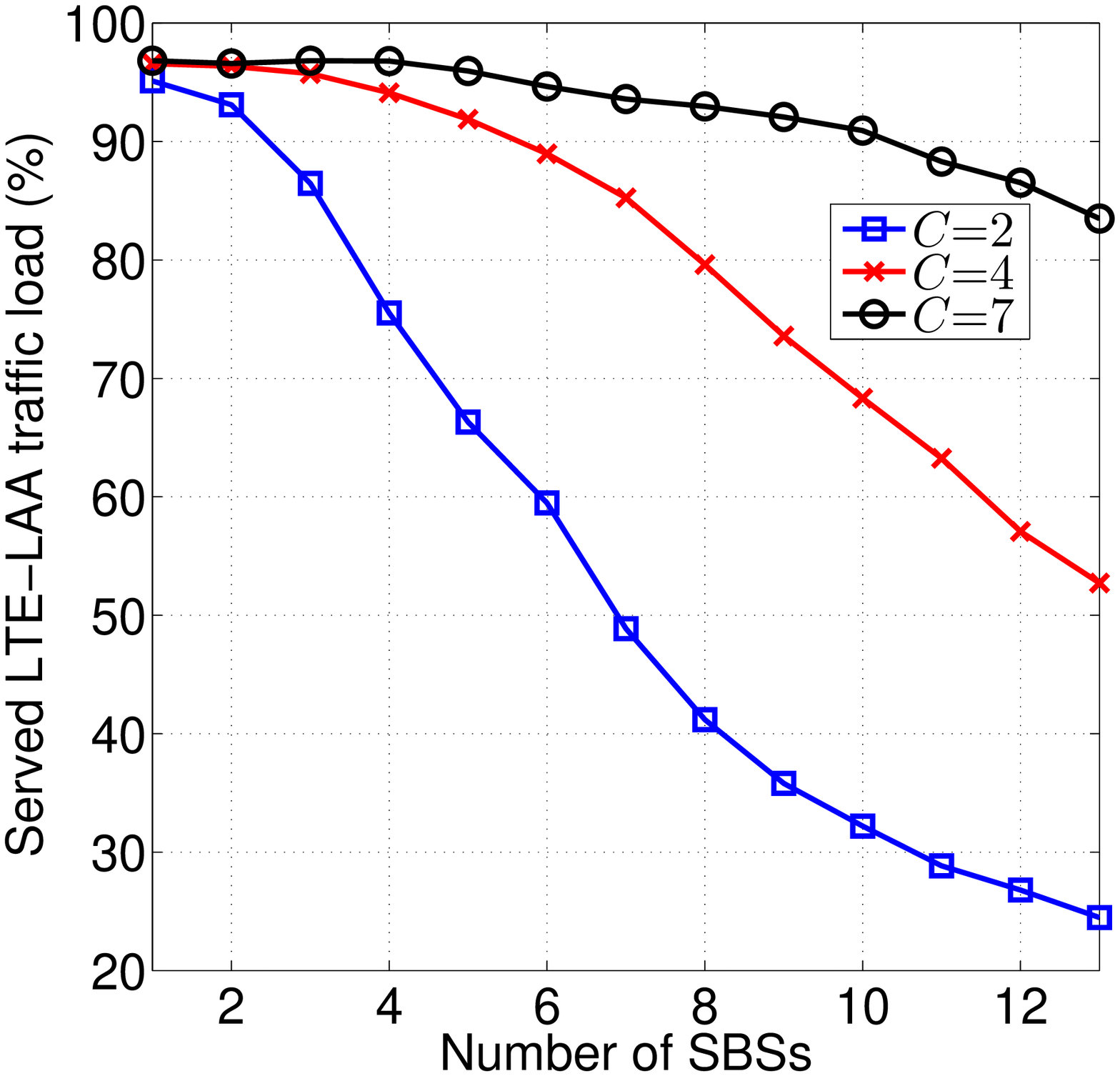}
   \caption{The proportion of LTE-LAA served traffic load as a function
    of the number of SBSs and for different number of unlicensed
    channels ($C=2$, $4$, and $7$).}\label{varying_LTE}
     \vspace{-1.2cm}
  \end{center}
\end{figure}

Fig.~\ref{varying_LTE} investigates the proportion of served LTE-LAA traffic for different network parameters.
From Fig.~\ref{varying_LTE}, we can see that, as the number of SBSs increases, the proportion of LTE-LAA served traffic, relative to its corresponding offered load decreases. Moreover, reducing the number of unlicensed channels leads to a decrease in the proportion of LTE-LAA served traffic. Although the number of available unlicensed channels are not players in the game, they affect spectrum allocation action selection for each SBS. As the number of channels increases, the action space for the channel selection vector increases, thus giving more opportunities for an SBS to serve more of its offered load.

Fig.~\ref{priority_ratio} shows the total network served traffic load as well as that of LTE-LAA and WiFi as a function of the priority fairness ratio on the unlicensed band $(P_\mathrm{{LTE}}$/$P_\mathrm{{WiFi}})$ for three different network scenarios considering $T=6$. From Figs.~\ref{priority_ratio} (b) and (c), we can see that more LTE-LAA and less WiFi traffic load is served as $P_\mathrm{{LTE}}$/$P_\mathrm{{WiFi}}$ increases and thus the priority fairness parameters $P_\mathrm{{LTE}}$ and $P_\mathrm{{WiFi}}$ can be regarded as network design parameters that can be adjusted in a way that would avoid LTE-LAA from aggressively offloading traffic to the unlicensed bands. Moreover, from Fig.~\ref{priority_ratio} (a), we can see that the served total network traffic load is maximized at $P_\mathrm{{LTE}}$/$P_\mathrm{{WiFi}}=1$ thus allowing an efficient utilization of the unlicensed spectrum. On the other hand, from Fig.~\ref{priority_ratio} (c), we can see that the served WiFi traffic load for our proposed scheme is greater than or equal to the served WiFi traffic load for the case in which LTE-LAA is replaced by an equivalent WiFi network for values of $P_\mathrm{{LTE}}$/$P_\mathrm{{WiFi}}$ less than $0.8$. From Fig.~\ref{priority_ratio} (c), we can conclude that the WiFi performance for our proposed spectrum sharing scheme, when considering equal weighted airtime share (i.e., $P_\mathrm{{LTE}}$/$P_\mathrm{{WiFi}}=1$), achieves very close performance to the case when only WLAN is operating over the unlicensed spectrum. For instance, the proportion of WiFi served traffic load corresponds to 68\% for the WiFi-LTE scenario as opposed to 70\% for the WiFi-WiFi scenario in the case of 4 SBS and 4 channels. This slight decrease is mainly due to the differences in the MAC layers of both technologies. For instance, LTE adopts a more efficient scheduling mechanism and has less overhead as compared to WiFi. In particular, the DCF protocol of WiFi results in the channel being unused for some period of time and, thus, WiFi should be given a slightly larger priority in that case. In summary, we can deduce that the values of $P_\mathrm{{LTE}}$ and $P_\mathrm{{WiFi}}$ can be regarded as tuning parameters that allow the network operator to achieve a tradeoff between efficiency and fairness.

\begin{figure*}[t!]
  \begin{center}
  \vspace{-0.6cm}
    \includegraphics[width=16cm]{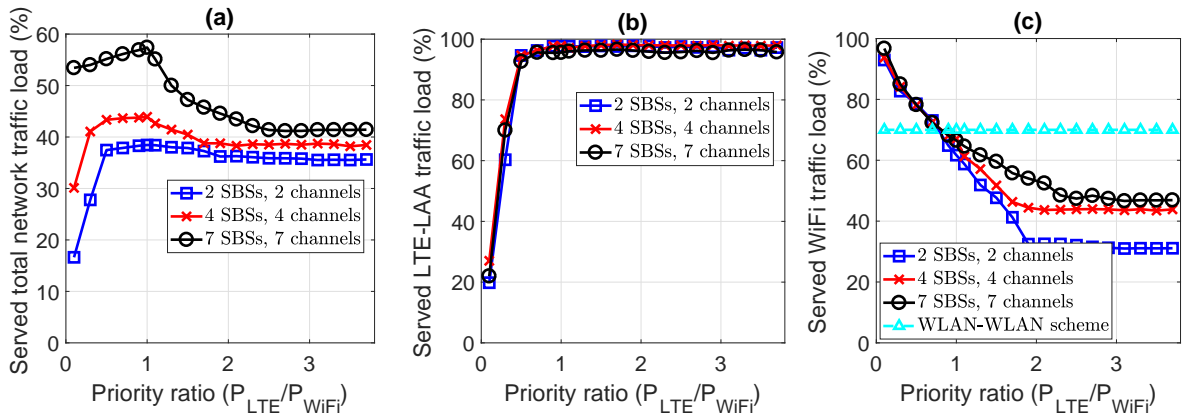}
     \vspace{-1cm}
   \caption{The proportion of the (a) total network served traffic load (b) LTE-LAA served traffic load and (c) WiFi served traffic load as a function of the priority fairness ratio on the unlicensed band, ($\mathrm{P_{LTE}}$/$\mathrm{P_{WiFi}}$). The straight line in (c) represents the proportion of WiFi served traffic load for the case when the LTE network is replaced by an equivalent WiFi network.}\label{priority_ratio}
     \vspace{-1.2cm}
  \end{center}
\end{figure*}







Fig.~\ref{convergence_learning_rate} shows the average value of airtime allocated to the LTE-LAA network as a function of the number of epochs required for the network to
converge while considering different values for the learning rate. The learning rate determines the step size the algorithm takes to reach the minimizer and
thus has an impact on the convergence rate of our proposed framework. Moreover, an epoch, which consists of multiple iterations, is a single pass through the entire training set, followed by testing of the verification set. From Fig.~\ref{convergence_learning_rate}, we can see that for $\lambda=0.1$, our proposed algorithm requires more than $50$ epochs to approximate the reward function, while, for $\lambda=0.01$, it only needs $20$ epochs.
In fact, for $\lambda=0.1$, we can see that our proposed algorithm fluctuates around a different region of the optimization space. Clearly, a learning rate that is too large can cause the algorithm to diverge from the optimal solution. This is because too large initial learning rates will decay the loss function faster and thus make the model get stuck at a particular region of the optimization space instead of better exploring it.
On the other hand, a learning rate that is too small results in a low speed of convergence. For instance, for $\lambda=0.0001$ and $\lambda=0.00005$, our proposed algorithm requires $\sim40$ epochs to converge. Therefore, although we use an adaptive learning rate approach, the optimization algorithm relies heavily on a good choice of an initial learning rate~\cite{learning_rate_initial}. In other words, the initial value of the learning rate should be within a particular range in order to have good performance. Choosing a proper learning rate is an important key aspect that has an impact on the solution as well as the convergence speed. The optimal value of the initial learning rate is dependent on the dataset under study, where for each dataset, there exists an interval of good learning rates at which the performance does not vary much~\cite{learning_rate}. This in turn necessitates the need for experimental studies in order to search for good problem-specific learning rates~\cite{learning_rate_initial}. A typical range of the learning rate for the dataset under study falls approximately between $0.0005$ and $0.01$, requiring $\sim20$ epochs.


\begin{figure}[t!]
  \begin{center}
  \vspace{-0.2cm}
    \includegraphics[width=11cm]{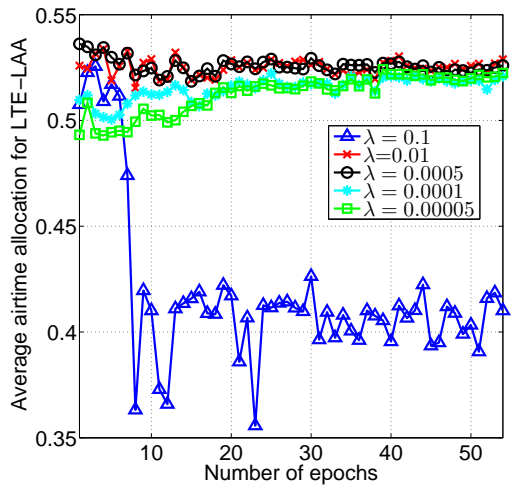}
   \caption{The average airtime allocated for LTE-LAA as a function of the number of epochs for different values of the learning rate.}\label{convergence_learning_rate}
     \vspace{-1.2cm}
  \end{center}
\end{figure}

\vspace{-0.35cm}
\section{Conclusion}
In this paper, we have proposed a novel resource allocation framework for the coexistence of LTE-LAA and WiFi in the unlicensed band. We have formulated a game model where each SBS seeks to maximize its rate over a given time horizon while achieving long-term equal weighted fairness with WLAN and other LTE-LAA operators transmitting on the same channel. To solve this problem, we have developed a novel deep learning algorithm based on LSTMs. The proposed algorithm enables each SBS to decide on its spectrum allocation scheme autonomously with limited information on the network state. Simulation results have shown that the proposed approach yields significant performance gains in terms of rate compared to conventional approaches that considers only instantaneous network parameters such as instantaneous equal weighted fairness, proportional fairness and total network throughput maximization. Results have also shown that our proposed scheme prevents disruption to WLAN operation in the case large number of LTE operators selfishly deploy LTE-LAA in the unlicensed spectrum. For future work, we will consider a multi-mode SBS that operates over the sub-6 GHz licensed band, sub-6 GHz unlicensed band, and the high-frequency mmWave band. Here, we aim at addressing the problem of coexistence of a multi-mode SBS with WiFi on the 5 GHz band and with WiGig on the mmWave band, simultaneously.
\vspace{-0.3cm}






\def\baselinestretch{0.98}
\bibliographystyle{IEEEtran}

\bibliography{references}
\vspace{-0.4cm}
\end{document}